\documentclass[journal,twoside,web]{ieeecolor}
\usepackage{generic}
\usepackage{cite}
\usepackage{amsmath,amssymb,amsfonts}
\usepackage{graphicx}
\usepackage{hyperref}
\usepackage{textcomp}
\usepackage{mathbbol}
\usepackage{dirtytalk}              
\usepackage{graphicx}
\usepackage{amsmath}
\usepackage{mathtools}
\usepackage{verbatim}
\usepackage{amssymb}

\newtheorem{remark}{Remark}
\newtheorem{thm}{Theorem}
\newtheorem{lemma}{Lemma}

\newtheorem{expm}{Example}

\DeclareMathOperator\Spec{Spec}
\usepackage[font=footnotesize,labelfont=bf]{caption}
\usepackage{subcaption}
\def\BibTeX{{\rm B\kern-.05em{\sc i\kern-.025em b}\kern-.08em
    T\kern-.1667em\lower.7ex\hbox{E}\kern-.125emX}}
\begin{document}
\title{\LARGE \bf
Absolute centrality in a signed Friedkin-Johnsen based model: a graphical characterisation of influence}
\author{Aashi Shrinate$^{1}$, \IEEEmembership{Student member, IEEE} and Twinkle Tripathy$^2$, \IEEEmembership{Senior Member, IEEE}
\thanks{$^{1}$Aashi Shrinate is a research scholar and $^{2}$ Twinkle Tripathy is an Assistant Professor in the Control and Automation specialization of the Department of Electrical Engineering, Indian Institute of Technology Kanpur, Kanpur, Uttar Pradesh, India, 208016. Email: {\tt\small aashis21@iitk.ac.in and ttripathy@iitk.ac.in}.}}
\maketitle
\begin{abstract}
This paper studies the evolution of opinions {governed by a} Friedkin-Johnsen (FJ) { based model in arbitrary network structures with signed interactions}. The agents contributing to the opinion formation are characterised as being \textit{influential}. Initially, { the agents are classified as opinion leaders and followers based on network connectivity and nature of interactions.} However, the addition of stubbornness leads to interesting behaviours wherein a \textit{non-influential} agent can now become \textit{influential} and \textit{vice-versa}.  Thereafter, a signal flow graph (SFG) based method is proposed to quantify the influence of an influential agent's opinions. { Additionally, it helps illustrate the role played by network topology in shaping the final opinions of the agents.} Based on this analysis, the \textit{absolute centrality measure} is proposed to determine the overall influence of all the agents in the network. Unlike most of the existing measures, it is applicable to any network structure and considers the effect of stubbornness and antagonism. Examples are presented throughout the paper to illustrate and validate these results.
\end{abstract}

\section{INTRODUCTION}
\IEEEPARstart{C}{onsider} a group of connected individuals communicating with each other. The opinion of an individual undergoes a natural transformation through the interactions within the group on topics of interest; this transformation depends on the importance it assigns to the opinions of those with whom it engages. 
The analysis of the resulting behaviours is a complex problem but an important one for our society, especially in this day and age when social networks are being employed to influence consumer behaviours \cite{voramontri2019impact}, voting preferences \cite{fernandez2014voter} and shaping public opinions via disinformation campaigns \cite{Gorodnichenko2021} among others.

Several models have been proposed to study opinion formation in a network of interacting agents, e.g. averaging based the DeGroot's model \cite{degroot1974reaching}, Friedkin-Johnsen (FJ) model \cite{FRIEDKIN1997209}, homophily-based Hegselmann-Krause \cite{rainer2002opinion}, Bayesian models \cite{10.2307/2118364}, \textit{etc}. However, despite their simplicity, experimental evidence has shown that the DeGroot-based models can sometimes predict opinion formation more accurately than the Bayesian models \cite{chandrasekhar2020testing}. { The FJ model, an extension of DeGroot's model, has become popular in the literature due to its simplicity and ability to explain diverse behaviours. Additionally, it is one of the few models that accurately predict individuals' opinions in human-subject experiments \cite{noah1999social}. }

The FJ model accounts for disagreement, the most commonly observed behaviour in a society, by introducing agent(s) in the network who are stubborn in their prejudices. 
Another peculiar aspect of such an opinion formation process is that the opinion value at which the convergence occurs often depends only on the initial states of certain agents in the network. We call an agent \textit{influential} if its initial opinion contributes to the final opinion of any agent in the network. 
This property is closely linked to the centrality measures, which generate a ranking of agents in a network based on their importance. Several popular measures exist, such as in-degree centrality, eigenvector centrality \cite{Bonacich}, Katz centrality \cite{katz1953new}, closeness centrality \cite{Bavelas1950} and betweenness centrality \cite{Freeman1977} that denote the relative importance of nodes in the network based on different graph properties. {  However, the authors in \cite{friedkin1991theoretical} note that a centrality measure applies to a social process only when it is derived from the latter.

The authors in \cite{friedkin1991theoretical} and \cite{Community_Cleavage} propose \textit{influence centrality} (IC) to quantify the contribution of each agent in the final opinion of the agents in the network when the opinions are governed the FJ model.
In \cite{proskurnikov2016pagerank}, the authors prove that IC becomes identical to the powerful PageRank centrality when all agents in the network are equally stubborn.} The evolution of social power (IC) in the FJ model over a sequence of issues is examined in \cite{Tian2022} under the assumption that each agent in the network displays a positive degree of stubbornness.
{  Note that the matrix inversion required to determine IC (equivalently the final opinions) is very inefficient for large networks. More importantly, 
this method fails to give any insight into the role of network interconnections on IC.

The authors in \cite{gionis2013opinion,GHADERI20143209} employ random walks on a network with absorbing nodes to determine the final opinions of agents governed by the FJ model. Additionally, the authors in \cite{GHADERI20143209} show the equivalence between the network of agents and an electric network with each stubborn agent acting as an ideal voltage source. The final opinion of each agent is the voltage at the corresponding node in the electric network. These approaches show the dependence of the final opinions of agents on both network interconnections and stubborn behaviour.  } 

The works in \cite{friedkin1991theoretical,Community_Cleavage,proskurnikov2016pagerank,Tian2022,GHADERI20143209,gionis2013opinion} determine the final opinions (equivalently, the IC) considering the network to have cooperative interactions which is admissible in various applications. Social networks, however, generally have competitive interactions along with cooperative. They are represented by signed networks.
The signed interactions in DeGroot's framework and Laplacian framework have been studied in \cite{xia2015structural} and \cite{altafini}, respectively. {  The signed interactions in the 
FJ model has only recently been introduced in \cite{altafini_an_FJ} over a single issue and a sequence of issues while considering each agent to have a positive degree of stubbornness. In \cite{zhou2024friedkin}, the authors determine the equilibrium solution of the FJ model with signed interactions by traversing the random walks on the augmented graph derived from the underlying network which is connected and undirected.}

{ In the aforementioned works relating to FJ model, each agent in the network is assumed to either be stubborn or have a path to a stubborn agent. As a result, only stubborn agents influence others and have a non-zero entry in the IC vector. However due to the heterogeneity among agents,} there may exist agents in the network which do not have a path to any stubborn agent. Such agents are referred to as \textit{oblivious agents} in the literature \cite{parsegov2016novel}. The convergence of opinions governed by the FJ model in the presence of oblivious agents is examined in { \cite{parsegov2016novel} and \cite{TIAN2018213}. Furthermore in \cite{TIAN2018213}, the necessary conditions for non-stubborn agents to be influential are presented. Note that in both of these works, the agent interactions are cooperative.}

In contrast to the works \cite{friedkin1991theoretical,proskurnikov2016pagerank,GHADERI20143209,gionis2013opinion,parsegov2016novel,TIAN2018213,altafini_an_FJ,zhou2024friedkin},
 this paper studies the opinion formation in signed networks with arbitrary connectivity using an FJ-based model { in the presence of oblivious agents.} In the given framework, we characterise the agents as opinion leaders and followers based on the network topology. {  Next, we determine the final opinion of the agents in the network and identify the influential agents in this process.} We propose an SFG-based methodology to precisely evaluate the values to which the opinions of the `influenced' agents converge and quantify the effect of the `influential' agents. { If the degree of stubbornness in the network changes}, we show that stubborn followers always become influential. More interestingly, we present the conditions under which some opinion leaders can now cease to remain \textit{influential}. In order to identify the most influential node(s) in the network, we also propose a new centrality measure derived from the opinion model, referred to as the \textit{absolute centrality measure}; its major advantages are that it is defined for any arbitrary network structure and {  it accounts for the effects of signed interactions, stubbornness and the opinion leaders.} 
 With this, we now highlight the major contributions of the work.

\begin{itemize}
    \item \textit{A generalised framework:} The opinion evolution model proposed in this work generalises the FJ and DeGroot models in \cite{parsegov2016novel} to a network with cooperative and antagonistic interactions. It is also a generalisation of the existing works on the quantification of influence as it is applicable to any network structure. Furthermore, we allow any agent in the network to be stubborn or non-stubborn, unlike the works in \cite{Tian2022,altafini_an_FJ}.
    
    { \item \textit{Convergence Analysis}: We analyse the convergence of opinions evolving by the FJ-based model, considering both the presence of signed interactions and the oblivious agents, unlike any other existing work. Based on the convergence analysis, we further determine the influential agents and derive the necessary and sufficient conditions an agent must satisfy to be influential.}
    \item \textit{A graphical representation of influence:} Our work presents a graphical approach to illustrate influence propagation in the network by using SFGs. To the best of our knowledge, this is one of the first works which uses SFGs for influence evaluation. Through this, we establish a direct correlation of influence with the network topology, the nature of interactions and stubborn behaviours.
    \item \textit{Absolute centrality measure:} Despite the popularity of FJ models, none of the existing influence centrality measures account for antagonism and stubborn behaviours simultaneously. 
    The \textit{absolute centrality measure} proposed in this work bridges this gap in the literature { since it is derived from the opinion evolution by FJ based model in signed networks. Additionally, this measure also accounts for the influence arising out of network properties.}  
\end{itemize}
The paper has been organised as follows: Sec. \ref{Sec2} discusses the required notations and the relevant preliminaries. 
The proposed opinion model and classification of the agents are presented in Sec. \ref{Sec3}. Sec. \ref{Sec4} analyses the convergence of opinions evolving by the proposed model. Sec. \ref{Sec5}
discusses the quantification of the influence of influential agents in any weakly connected networks by the SFGs with illustrative examples. The absolute influence centrality is presented in Sec. \ref{Sec6}. We conclude the paper in Sec. \ref{sec:con} with some insights into the possible future research directions.
\section{Notations \& Preliminaries}
\label{Sec2}
\subsection{Notations}
\label{subsec:NOT}
The vector $\mathbb{1}_n \in \mathbb{R}^n$ denotes a column vector with all entries equal to $+1$. The identity matrix of dimension $n$ is denoted by $I_n$. We use $\mathbb{0}$ to denote a matrix with all entries equal to $0$ of appropriate dimensions. For a given matrix $M =[m_{ij}] \in \mathbb{R}^{n \times n}$, let $\tilde{M}=[|m_{ij}|]$, where $|m_{ij}|$ is the absolute value of $(i,j)^{th}$ entry of $M$ for all $i,j=\{1,2,...,n\}$. A diagonal matrix $M \in \mathbb{R}^{n \times n}$ is denoted by $M=diag([m_1,m_2,...,m_n])$. 
The spectrum of $M$ is the set of all eigenvalues of $M$, denoted by $\Spec(M)$. The spectral radius is the maximum norm of the eigenvalues of matrix $M$, denoted by $\rho(M)=\max\{ |\lambda| : \lambda \in \Spec(M)\} $.
\subsection{Graph Preliminaries}
\label{subsec:GP}
A graph is defined as $\mathcal{G}=\{\mathcal{V},\mathcal{E},A\}$ where $\mathcal{V}=\{1,2,...,$ $n\}$ is the set of nodes representing the $n$ agents in the network, $\mathcal{E} \subseteq \mathcal{V} \times \mathcal{V}$ is the set of ordered pair of nodes called edges which denote the communication topology of the network and $A=[a_{ij}] \in \mathbb{R}^{n \times n}$ is the signed weighted adjacency matrix with $a_{ij} \neq 0$ if and only if $(i,j) \in \mathcal{E}$. The ordered pair $(i,j) \in \mathcal{E}$ implies that node $i$ has a directed edge to node $j$. Additionally, it informs that node $i$ is an in-neighbour of node $j$ and node $j$ is an out-neighbour of node $i$. The entry $a_{ij}$ denotes the edge weight of edge $(i,j) \in \mathcal{E}$. 

A graph is undirected if $(i,j) \in \mathcal{E}$ implies that $(j,i) \in \mathcal{E}$. On the other hand, the ordered pair $(i,j) \in \mathcal{E}$ in a directed graph (digraph) does not imply that $(j,i) \in \mathcal{E}$. The neighbourhood of an agent $i$, composed of its out-neighbours, is defined as $N_i=\{j \in \mathcal{V}|(i,j) \in \mathcal{E}\}$.
A sink is a node in $\mathcal{G}$ without any out-neighbours and a source is a node without any in-neighbours.

A path is an ordered sequence of nodes in which every pair of adjacent nodes produces an edge that belongs to the set $\mathcal{E}$. A cycle is a path that starts and ends at the same node. A digraph is acyclic if it does not contain any cycles. An undirected graph is connected if there exists a path between every pair of nodes. A digraph is a strongly connected graph if a directed path exists between every pair of nodes in the graph. A graph is weakly connected if it is not strongly connected but its undirected version is connected. 

The condensation graph of a graph $\mathcal{G}$ is defined as $C(\mathcal{G})=(\mathcal{V}_c,\mathcal{E}_c)$. Each node $\mathcal{I} \in \mathcal{V}_c$ is a {   strongly connected component (SCC)} of graph $\mathcal{G}$, and an edge $(\mathcal{I},\mathcal{J})\in \mathcal{E}_c $ exists if and only if an edge $(i,j)\in \mathcal{E}$ exists in graph $\mathcal{G}$ from node $i \in \mathcal{I}$ to a node $j \in \mathcal{J}$.  A sink of the condensation graph is {an   SCC} that forms a node in the $C(\mathcal{G})$ without any outgoing edges.

A graph is structurally balanced if there exists a bipartition of vertices $\mathcal{V} \text{ such that } \mathcal{V}_1 \cap \mathcal{V}_2=\emptyset$ and $\mathcal{V}_1 \cup \mathcal{V}_2 =\mathcal{V}$ 
with positive interaction $a_{ij} \geq 0$ between nodes  $i,j \in \mathcal{V}_q \ (q\in\{1,2\})$ and negative interaction $a_{ij} \leq 0$ if  $i \in \mathcal{V}_{p}$ and $j \in \mathcal{V}_q,p\neq q, (p,q\in\{1,2\})$. A graph that is not structurally balanced is structurally unbalanced.
\subsection{Signal Flow Graph Preliminaries}
\label{subsec:SFGP}
{  An SFG $\mathcal{G}_s=({V}_s,B_s)$ is a digraph formed by nodes $\{1,...,k\}$ in the set ${V}_s$ and branches in the set ${B}_s$. Each node $i$ in $\mathcal{G}_s$ is associated with a node signal or state $y_i$. An ordered pair $(j,i)$ denotes a branch from node $j$ to node $i$ that has an associated branch gain $g_{i,j} \in \mathbb{R}$. A branch that originates from $j$ and terminates at $i$ represents the dependence of $y_i$ on $y_j$. The relation among the node states is described by a set of linear equations of the form,
\begin{align}
\label{SFG_description}
   y_i=\sum_{j} g_{i,j} y_j \qquad i,j \in \{1,2,...,k\}.
\end{align}
It follows from eqn. \eqref{SFG_description} that a state $y_i$ can be determined using the state $y_j$ of node $j$ from which it has an incoming branch. We now define some important terms associated with $\mathcal{G}_s$.

\begin{itemize}
    \item In the context of an SFG, a sink node is one which does not affect any of the other node states. On the other hand, a source node is one which does not get affected by any other node state.
    \item A \textit{forward path} is one where none of the nodes is traversed more than once. A \textit{feedback loop} (or a cycle) is a forward path where the first and the last nodes coincide. 
    \item A pair of loops in $\mathcal{G}_s$ are \textit{non-touching} if none of the nodes or branches in the loops are common. Similarly, a subgraph of $\mathcal{G}_s$ does not touch a forward path if it contains all the nodes and branches of $\mathcal{G}_s$ except those occurring in the forward path.
  \item The \textit{path gain} of a path and the \textit{loop gain} of a loop in $\mathcal{G}_s$ is the product of the branch gains of the branches forming the path and the loop, respectively.
\end{itemize}

The gain of an SFG $\mathcal{G}_s$ is the signal obtained at a sink for a unit signal at a source. Given any pair of sink and source nodes, the gain is determined using the Mason's gain formula \cite{mason1956feedback} given as follows,
\begin{equation}
\label{eq:sfg-gain}
    G=\frac{\sum_{h}G_h\Delta_h}{\Delta}
\end{equation}
where $G_h$ is the gain of the $h^{th}$ forward path from the given source to the sink and,
\begin{align*}
    \Delta=1-\sum_{m}P_{m1}+\sum_{m}P_{m2}-\sum_{m}P_{m3}+...
\end{align*}
with $\mathcal{P}_{mr}$ being the product of loop gains of $m^{th}$ combination of $r$ number of non touching loops. $\Delta_h$ is $\Delta$ of the subgraph of $\mathcal{G}_s$ not touching the $h^{th}$ forward path. }
\subsection{Matrix Preliminaries}
\label{subsec:MP}
A matrix $M\in \mathbb{R}^{n \times n}$ is row stochastic if $M$ is non-negative and $M\mathbb{1}_n=\mathbb{1}_n$. It is row substochastic if it is non-negative and all the row sums are at most one and at least one of the row sums is less than one. {  $M \in \mathbb{R}^{n \times n}$ is semi-convergent if the $\lim_{k \to \infty}M^k$ exists. It is convergent if $\lim_{k \to \infty}M^k=\mathbb{0}$. The spectrum of a semi-convergent and not convergent matrix must have the following properties:
\begin{itemize}
    \item $1$ is a simple or semi-simple eigenvalue,
    \item all the other eigenvalues have magnitude less than $1$.
\end{itemize}}The following results will be used in the paper.
\begin{lemma}[\cite{gelfand1941normierte}]
\label{lm:1}
For any $M \in \mathbb{C}^{n \times n}$ and any induced norm $\|.\|$, the Gelfand's formula states { that the spectral radius} $\rho(M)$ of matrix $M$ is  $\rho(M)=\lim_{k \to \infty}\|M^k\|^{\frac{1}{k}}$.
\end{lemma}

\begin{lemma}
\label{lm:2}
The spectral radius of a matrix $M \in \mathbb{R}^{n \times n}$ and $\Tilde{M}=[|m_{ij}|]$, derived from $M$, have the following relation.
\begin{align*}
    \rho(M)\leq \rho(\Tilde{M})
\end{align*}
\end{lemma}
\begin{proof}
    Since row sum of $M^k$ is strictly less or equal to row sum of $\tilde{M}^k$ we get $\|M^k\|_{\infty} \leq \| \Tilde{M}^k\|_{\infty}$, which implies $\lim_{k \to \infty}\|M^k\|_{\infty}^{1/k} \leq \| \lim_{k \to \infty}\Tilde{M}^{k}\|_{\infty}^{1/k}$. { Using} the Gelfand's formula from Lemma \ref{lm:1}, we get $\rho(M) \leq \rho(\Tilde{M})$.
\end{proof}
\section{Opinion Dynamics}
\label{Sec3}
In a society, more often than not, there are individuals or groups of individuals who have high credibility and expertise in certain domains or wield significant power over the thoughts, beliefs, and actions of a broad audience. They can act as opinion leaders of the group who play a major role in opinion formation. {  Some examples of real-world opinion leaders include socio-political leaders, successful entrepreneurs and highly reputed scholars who have a significant impact on society}.
The advent of social media has led to the emergence of a new kind of opinion leader, referred to as \textit{influencer}, on whom individuals rely directly or indirectly for news and content. As a result, these influencers play a significant role in the formation of public opinion and affect consumer behaviour \cite{voramontri2019impact}, political ideologies \cite{ding2023electoral}, stock prices \cite{dougan2020speculator}, \textit{etc}.

It is also important to note that despite the widespread influence of opinion leaders, certain outliers may exist who resist changes in their perception and do not readily accept the influence of opinion leaders. We refer to them as \textit{stubborn} agents. There is a limited understanding of how opinion leaders influence the opinions of individuals with varying degrees of stubbornness. Our objective is to study and quantify the \textit{influence} that individuals in a network exert on each other when their opinions evolve in a heterogeneous group.

\subsection{FJ based opinion model}
\label{subsec:SFJM}
In this paper, we study the evolution of opinions in a network of agents where certain agents that are stubborn with respect to their prejudices.
{ The FJ model, an extension of DeGroot's model, considered the existence of such agents who take a convex combination of their neighbours' opinions and their own prejudices to update their opinions in a cooperative network.
However, individuals interacting in a group when discussing issues of significant importance such as politics \cite{NEAL2020103}, international relations \cite{harary1961structural,moore1978international}, sports \cite{fink2009off} \textit{etc.} may have competing interests which are denoted by antagonistic interactions.}

Consider a signed network $\mathcal{G}$ with agents indexed $1$ to $n$. The opinions of agents at the $k^{th}$ instance are given by the vector $\mathbf{x}(k)=[x_1(k),x_2(k),...,x_n(k)] \ \in \mathbb{R}^n$ where $x_i(k)$ is the opinion of the $i^{th}$ agent in the group. We assume that $\mathcal{G}$ has no self-loops.
The following discrete-time opinion dynamics model explains the opinion formation process in a group of $n$ heterogeneous agents, 
 \begin{equation}
    x_i(k+1)=\gamma_ix_i(k)+
    \beta_ix_i(0)+(1-\gamma_i-
    \beta_i)\sum_{j =1}^{n}q_{ij}x_j(k)
    \label{eq:opinion_dynamics}
\end{equation} 
where $\gamma_i \in [0,1]$ models the self-belief agent $i$ has in its opinion at $k^{th}$ instance, $\beta_i \in [0,1)$ denotes the degree of stubbornness of agent $i$ towards its initial opinion. 
The matrix $Q=[q_{ij}]$ is defined as:
\begin{align}
\label{eq:weighted_Adjacency}
q_{ij}=\begin{cases}
    \frac{a_{ij}}{\sum_{j=1}^{n}|a_{ij}|} & \text{ if }\sum_{j=1}^{n}|a_{ij}| \neq 0  \\
   1 & \text{ if }\sum_{j=1}^{n}|a_{ij}| = 0 \text{ and } i=j\\
   0 & \text{ if }\sum_{j=1}^{n}|a_{ij}| = 0 \text{ and } i\neq j
\end{cases}    
\end{align}
The opinion model parameters are chosen such that a node $i$, which is not a sink in the network $\mathcal{G}$, satisfies $\beta_i+\gamma_i < 1$ to ensure network structure is not altered by model parameters. { An agent with $\beta_i>0$ is a stubborn agent and the set of all the stubborn agents in $\mathcal{G}$ is denoted by $\mathcal{V}_{S}$.}

\begin{remark}
{ A fully stubborn agent ($\beta_i=1$) and a sink in the network $\mathcal{G}$ are equivalent in the sense that an agent's opinion remains unchanged in either of the cases. An agent which is not a sink of $\mathcal{G}$ has $\beta_i<1$, so $\beta_i \in [0,1)$ for all $i\in \mathcal{V}$.}
\end{remark}


Finally, the opinion dynamics model in vector form is given by,
\begin{equation}
\label{vector_op_model}
    \mathbf{x}(k+1)=\big(\Gamma+(I-\Gamma-\beta)Q\big)\mathbf{x}(k)+\beta \mathbf{x}(0)
\end{equation}
where $\mathbf{x}(0)$ is the initial opinion vector, the matrix $\beta=diag([\beta_1,\beta_2,...,\beta_n])$ is a diagonal matrix.
We are interested in the steady-state behaviours of the opinions arising from the proposed model. To ease this analysis, we propose a classification of the agents as discussed next. 
\subsection{Classification of agents}
\label{subsec:COA}
{ Social networks, in general, are weakly connected with strongly connected subgroups of individuals formed on the basis of shared interests, geographical locations, culture, \textit{etc.} Considering these strongly connected subgroups as nodes, the condensation graph $C(\mathcal{G})=(\mathcal{V}_c,\mathcal{E}_c)$ is derived from the network $\mathcal{G}$. For a weakly connected network $\mathcal{G}$, its $C(\mathcal{G})$ is a directed acyclic graph which comprises one or more sinks.} We define $\mathcal{S}$ as the set comprising of the sinks of $C(\mathcal{G})$.
   \begin{align}
   \mathcal{S} \coloneqq \{S_i|S_i \in \mathcal{V}_c, N_{S_i}=\emptyset\}
   \end{align}
where $S_i$ represents the $i^{th}$ sink of $C(\mathcal{G})$ 
and $N_{S_i}$ denotes the set of out-neighbours of node $S_i$ in $C(\mathcal{G})$. The number of sinks in graph $C(\mathcal{G})$ is denoted by $n_s=|\mathcal{S}|$. Note that the $i^{th}$ sink $S_i$ of $C(\mathcal{G})$ is also set of node(s) in graph $\mathcal{G}$.
\begin{expm}
   Let us consider a network of agents represented by $\mathcal{G}=(\mathcal{V},\mathcal{E},A)$ in Fig. \ref{fig_1_1} and with its condensation graph $C(\mathcal{G})=(\mathcal{V}_c,\mathcal{E}_c)$ in Fig. \ref{fig_1_2}. The network $\mathcal{G}$ is weakly connected and has sinks $\mathcal{S}=\{S_1,S_2,..,S_5\}=\{ \{5,6,7\},$ $
\{8,9,10\},\{11\},\{12,13,14\},\{15,16,17\}\}$. 
\end{expm} 
\begin{figure}[h]
\centering
  \begin{subfigure}{0.3\textwidth}
  \centering
    \includegraphics[width=0.9\linewidth]{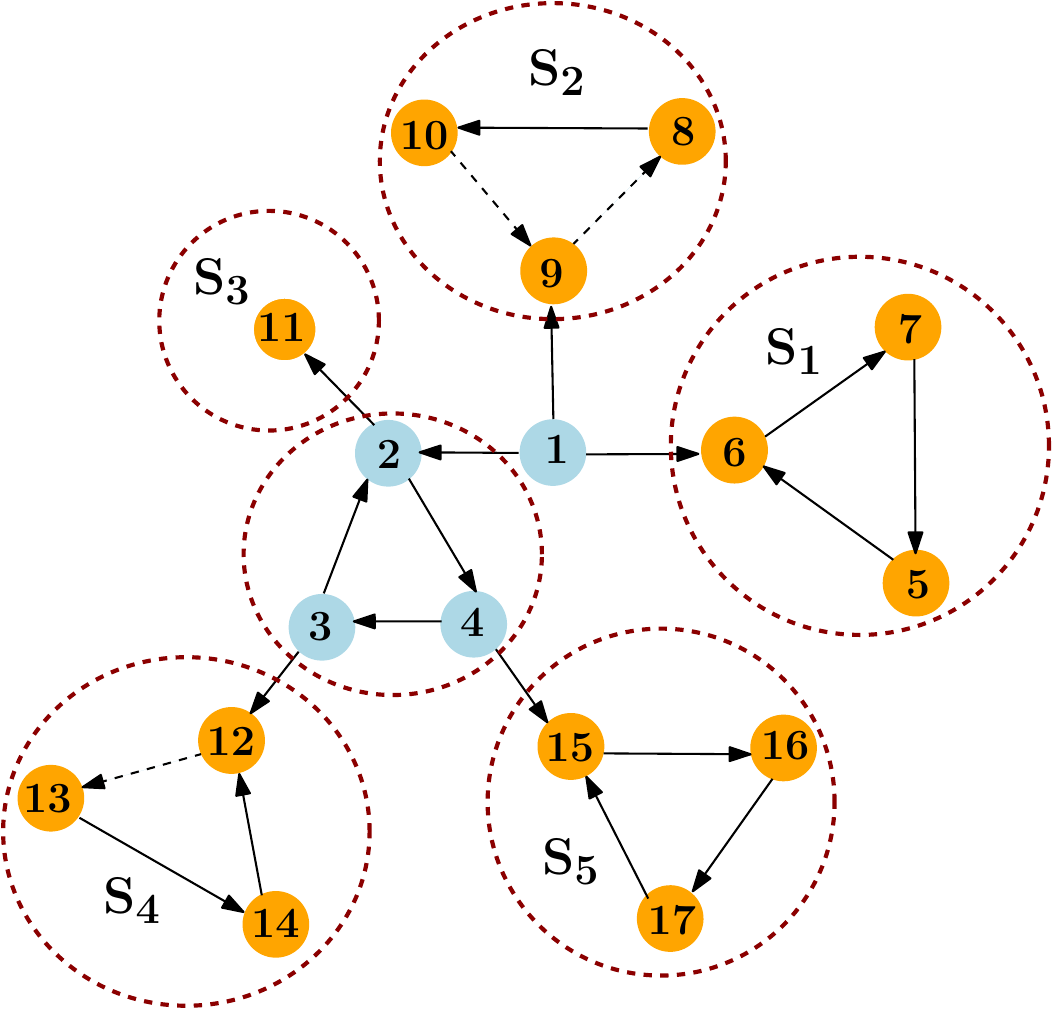}
    \caption{A weakly connected network $\mathcal{G}$ with circled \\ SCCs.}
    \label{fig_1_1}
  \end{subfigure}%
  \begin{subfigure}{0.2\textwidth}
  \centering
    \includegraphics[width=0.9\linewidth]{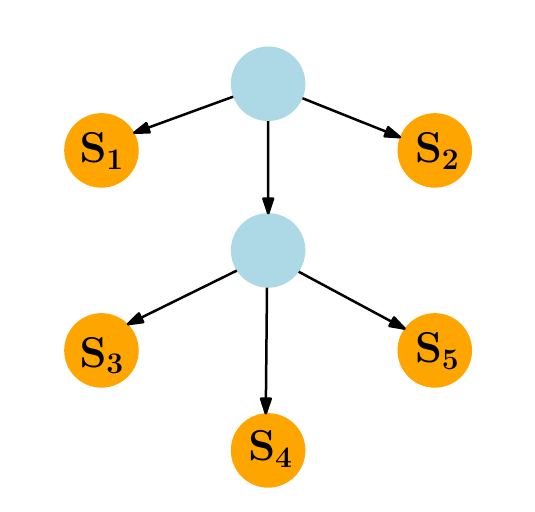}
    \caption{The condensation graph $C(\mathcal{G})$ of network $\mathcal{G}$}
     \label{fig_1_2}
  \end{subfigure}
  \caption{The followers and the opinion leaders of $\mathcal{G}$ are represented by blue and orange nodes, respectively. The dotted edges denote antagonistic interactions.}
\end{figure}
   
This subsection focuses on arriving at a mathematical classification of the agents in a network regarding their ability to influence the other agents. We begin with the identification of the opinion leaders in the network.
{  In a network $\mathcal{G}$ with opinions evolving according to eqn. \eqref{eq:opinion_dynamics}, an agent which belongs to a sink $S_i$ of $C(\mathcal{G})$ interacts only with the other agents in $S_i$. As a result, its opinion remains unaffected by the agents in $\mathcal{G}$ not belonging to $S_i$. On the other hand, a node in $\mathcal{G}$ not belonging to any sink of $C(\mathcal{G})$ must have an edge or a directed path to the node(s) belonging to a sink of $C(\mathcal{G})$. This occurs because $C(\mathcal{G})$ is a directed acyclic graph, Thus, the latter's opinion depends on the former during opinion evolution. This results in the following classification of agents:}
   \begin{itemize}
    \item An opinion leader is a node of network $\mathcal{G}$ belonging to a sink $S_i$ of  $C(\mathcal{G})$. The set of all the opinion leaders is defined as, 
    \begin{equation}
        \mathcal{V}_{o}\coloneqq \{j | \ j \in S_i \text{ and } i\in\{1,2,..,n_s\}\}.
    \end{equation}
    \item The set of followers $\mathcal{V}_{F}$ contains the nodes in the network $\mathcal{G}$ that do not belong to any sink of $C(\mathcal{G})$.  $$\mathcal{V}_{F} \coloneqq \mathcal{V}\setminus \mathcal{V}_{o}$$
\end{itemize}
 Each opinion leader in the network $\mathcal{G}$ is associated with a sink $S_i \in \mathcal{S}$ of $C(\mathcal{G})$. Thus, we define a function $\mathcal{S}^v(\cdot):\mathcal{V}\to \mathcal{S}$, which takes an opinion leader as input and returns the sink of $C(\mathcal{G})$ from set $\mathcal{S}$ to which the former belongs.
     \begin{align}
         \mathcal{S}^v(j)=\begin{cases}
             S_i & j \in \mathcal{V}_o \\
             \emptyset & j \in \mathcal{V}_F
         \end{cases}
     \end{align} 


Two types of opinion leaders in $\mathcal{G}$ can exist in the network.
     \begin{itemize}
    \item An opinion leader which is the only node in a sink of $C(\mathcal{G})$. The set of all such nodes of $\mathcal{G}$ is denoted as:
    \begin{equation}
        \mathcal{V}_{o1} \coloneqq \{j | \ j\in \mathcal{V}_o  , |\mathcal{S}^v(j)|=1\}
    \end{equation}
    where $|\mathcal{S}^v(j)|$ is the cardinality of sink $\mathcal{S}^v(j)$ equal to the number of opinion leaders belonging to the sink $\mathcal{S}^v(j)$. { Additionally, we know that every $i \in \mathcal{V}_{o1}$ is a sink in $\mathcal{G}$ and from eqn. \eqref{eq:opinion_dynamics} it follows that $i'$s opinion remains at its initial opinion irrespective of the magnitudes of $\beta_i$ and $\gamma_i$.    }
    \item An opinion leader that belongs to a sink of $C(\mathcal{G})$ consisting of two or more opinion leaders. The set of all such nodes of $\mathcal{G}$ is denoted as:
    \begin{equation}
        \mathcal{V}_{o2} \coloneqq \{j | \ j \in \mathcal{V}_o , |\mathcal{S}^v(j)|>1\}
    \end{equation}
    { Since opinion leaders in $\mathcal{V}_{o2}$ govern the process of opinion formation while still interacting with other agents, we consider them to be necessarily self-confident i.e. $\gamma_i>0$ for all $i \in \mathcal{V}_{o2}$. This has been established for DeGroot's model in \cite{PengJia_Reducible}. }
\end{itemize}
This classification allows us to analyse the effects of different kinds of agents on the opinion formation process. Note that an opinion leader in $\mathcal{V}_{o1}$ does not interact with the other agents in the network, whereas an opinion leader in $\mathcal{V}_{o2}$ belonging to a sink $S_i$ interacts only with other opinion leaders belonging to $S_i$. A follower in set $\mathcal{V}_{F}$ interacts with other followers or opinion leaders to update its opinion.  

In Fig. \ref{fig_1_1}, the opinion leaders in $\mathcal{G}$ are $\mathcal{V}_o$ are $\{5,6,...,17\}$ and the rest are followers. The opinion leader in $\mathcal{V}_{o1}=\{11\}$ is associated with sink $\mathcal{S}^v(11)=S_3$ of $C(\mathcal{G})$ and opinion leaders in $\mathcal{V}_{o2}=\{5,6,7,8,9,10,12,13,14,15,16,17\}$ are associated with sinks $\mathcal{S}^v(5)=\mathcal{S}^v(6)=\mathcal{S}^v(7)=S_1$, $\mathcal{S}^v(8)=\mathcal{S}^v(9)=\mathcal{S}^v(10)=S_2$, $\mathcal{S}^v(11)=S_3$, $\mathcal{S}^v(12)=\mathcal{S}^v(13)=\mathcal{S}^v(14)=S_4$ and $\mathcal{S}^v(15)=\mathcal{S}^v(16)=\mathcal{S}^v(17)=S_5$, respectively.

The sinks of $C(\mathcal{G})$ containing two or more opinion leaders from set $\mathcal{V}_{o2}$ can be further classified based on the nature of their interactions, as explained below.
\begin{itemize}
    \item $\mathcal{S}_{cp}$ is the set of sinks of $C(\mathcal{G})$ such that the associated opinion leaders in a sink $S_i \in \mathcal{S}_{cp} $ in $\mathcal{G}$ have cooperative interactions, equivalently, 
    \begin{align}
    \mathcal{S}_{cp}:=\{S_i| S_i\in \mathcal{S}, a_{jk}\geq 0 \ \forall \ j,k \in S_i \}.
    \label{eq:nc}
    \end{align}
    \item $\mathcal{S}_{bal}$ is the set of sinks of $C(\mathcal{G})$ such that the associated opinion leaders in a sink $S_i \in \mathcal{S}_{bal} $ have at least one antagonistic interaction and form a structurally balanced subgraph in $\mathcal{G}$. Then,
    \begin{align}
      \mathcal{S}_{bal}=\{S_i|S_i \in \mathcal{S},~&\exists \ j,k \in S_i \text{ such that } a_{jk}<0,\nonumber\\
      &S_i \text{ is structurally balanced}\}
      \label{eq:nbal}
    \end{align}
    \item  $\mathcal{S}_{unbal}$ is the set of sinks of $C(\mathcal{G})$ such that the associated opinion leaders in $\mathcal{G}$ form a structurally unbalanced subgraph. Then, 
    \begin{align}
    \mathcal{S}_{unbal}=\{S_i|S_i &\in \mathcal{S}, ~\exists \ j,k \in S_i \text{ such that } a_{jk}<0,\nonumber\\
    &S_i\text{ is structurally unbalanced} \}.
    \label{eq:nunbal}
    \end{align}
\end{itemize}
Structural balance property implies that the relations among agents in a group satisfy Heider's Laws \cite{heider1946attitudes}. 
 In general, even agents with cooperative interactions are structurally balanced \cite{altafini}. This implies that the sinks associated with opinion leaders in $\mathcal{V}_{o1}$ and the sinks in $\mathcal{S}_{cp}$ are trivially structurally balanced. We define the set of all sinks of $C(\mathcal{G})$, whose associated nodes in $\mathcal{G}$ form an effectively structurally balanced subgraph as,
\begin{align}
   S_b:=\cup_{j\in \mathcal{V}_{o1}}\mathcal{S}^v(j) \cup \mathcal{S}_{cp} \cup \mathcal{S}_{bal}.
   \label{eq:sb}
\end{align}
The opinion leaders in a sink belonging to $\mathcal{S}_{unbal}$ form a strongly connected and structurally unbalanced subgraph. The opinion leaders in the sink do not satisfy Heider's Laws resulting in cognitive dissonance \cite{festinger1957theory}.

%
Finally, given the presence of stubborn agents in the network, $\mathcal{V}_{S}$ comprises of the stubborn followers in $\mathcal{V}_{F}$ and stubborn opinion leaders in $\mathcal{V}_{o2}$. Since stubbornness may arise even in a sink of $C(\mathcal{G})$ composed of two or more opinion leaders, we define a set $\mathcal{S}_n$ to distinguish the sinks in $C(\mathcal{G})$ which consist of non-stubborn opinion leaders and form an effectively structurally balanced subgraph. $\mathcal{S}_n$ can be expressed as:
\begin{align}
   \mathcal{S}_n \coloneqq \{S_i\in S_b: \forall j \in S_i,\ j \in \mathcal{V}\setminus \mathcal{V}_{S}\} 
\end{align}
The classification of agents discussed in this section allows us to analyse the effect of different kinds of opinion leaders in the network.

\section{Convergence analysis}
    \label{Sec4}
In this section, we study the convergence of opinions of agents evolving according to eqn. \eqref{eq:opinion_dynamics} in a weakly connected signed network. The nodes of a weakly connected graph can be suitably permuted such that the adjacency matrix becomes block triangular. Therefore, we renumber the nodes such that $i=\{1,2,...,m\}$ are the followers and the rest are opinion leaders, where the opinion leaders associated with a sink of $C(\mathcal{G})$ are grouped together resulting in block triangular $A$. { Henceforth, we will use this numbering of nodes throughout the paper.} Further, we define the matrix $P$ as $P=\big(\Gamma+(I-\Gamma-\beta)Q\big)$. By definition, $P$ is block triangular of the following form,
\begin{align}
\label{B_rearranged}
P=\begin{bmatrix}
 P_{11} & P_{12} & ...   & P_{1(n_s+1)} \\
 \mathbb{0} &  P_{22} & \mathbb{0} & \mathbb{0}\\
  \vdots & ...  & \ddots & \vdots \\
 \mathbb{0} & \mathbb{0} & \mathbb{0} & P_{(n_s+1)(n_s+1)} \\
\end{bmatrix}
\end{align}
where $P_{ij}$ are submatrices for $i,j\in \{1,...,n_s+1\}$. The network is weakly connected; hence, each follower has a path to at least one of the opinion leaders. Thus, 
\begin{align}
\label{B:connected}
    P_{1j} \neq \mathbb{0}~\forall \ j \in \{2,...,n_s+1\}.
\end{align}
Since $P$ gets transformed to the block triangular form as shown in eqn. \eqref{B_rearranged}, it gets endowed with interesting spectral properties as given below.
\begin{thm}
\label{lm:6}
 $P$ is semi-convergent { (and not convergent)} if and only if $\mathcal{S}_n$ is non-empty, otherwise it is convergent.
\end{thm}
\begin{proof}
We know that $P$ is block triangular, hence, $\Spec(P)=\cup_{i=1}^{(n_s+1)}\Spec(P_{ii})$. So, we analyse the spectrum of each submatrix $P_{ii}$ for $i \in \{1,...,n_s+1\}$. Consider the submatrix $P_{11}$ associated with followers in the network. We define the matrix $\tilde{P}_{11}$ derived from $P_{11}$ as $\tilde{P}_{11}=[|p_{ij}|]$. In order to determine the spectral properties of $P_{11}$, we analyse the spectrum of the associated non-negative matrix $\tilde{P}_{11}$.

When at least one follower in the network is stubborn, the following scenarios may occur:
\begin{itemize}
    \item[(a)] If the follower $i$ is not stubborn and all of its neighbours are followers, then the row-sum for follower $i$ in $\tilde{P}_{11}$ is equal $\sum_{j\in \mathcal{V}_{F}}|p_{ij}|=1$.
    \item[(b)] If the follower $i$ is stubborn or some of its neighbours are opinion leaders, or both, then the row sum for follower $i$ in $\tilde{P}_{11}$ is $\sum_{j\in \mathcal{V}_{F}}|p_{ij}|<1$.
\end{itemize}
When none of the followers satisfy condition $(a)$, then it directly follows that $\rho(\Tilde{P}_{11})<1$ as $\rho(\Tilde{P}_{11})\leq \|\Tilde{P}_{11}\|_{\infty}<1$. If there exists one or more followers that satisfy condition $(a)$, then $\Tilde{P}_{11}$ is row substochastic. Since the network is weakly connected, each follower has a directed path to one or more opinion leaders. Thus, every follower with a row-sum equal to $1$ in $\tilde{P}_{11}$ has a path to a follower $j$ whose neighbour is an opinion leader. This implies that $\rho(\Tilde{P}_{ii})<1$ \cite{FB-LNS}. Consequently, from Lemma \ref{lm:2}, we infer that $\rho(P_{11})<1$. 
Similar arguments hold even if none of the followers in the network are stubborn resulting in $\rho(P_{11})<1$.

 Now, we discuss the spectral properties of the blocks $P_{ii}$ associated with the opinion leaders. An opinion leader $j$ in $\mathcal{V}_{o1}$ in sink $S_i$ of $C(\mathcal{G})$ is associated with submatrix $P_{(i+1)(i+1)}=[1]$. This implies that it contributes a simple eigenvalue $1$ to the spectrum of $P$. On the other hand, an opinion leader $j$ in $\mathcal{V}_{o2}$ belongs to a sink $S_i$ is associated with the square submatrix $P_{(i+1)(i+1)}$ of dimension strictly greater than one. If none of the opinion leaders in $\mathcal{S}_i$ are stubborn, then the following scenarios can arise.
        \begin{itemize}
            \item Suppose opinion leaders in $S_i$ have cooperative interactions such that $S_i\in \mathcal{S}_{cp}$. Owing to strong connectedness, { $P_{ii}$} is both row-stochastic and primitive. It implies that { $P_{ii}$} has a simple eigenvalue $1$ in its spectrum and all other eigenvalues have magnitudes less than one.
            \item Suppose the opinion leaders associated with the sink $S_i$ have antagonistic interactions and form a structurally balanced and strongly connected subgraph implying $S_i\in \mathcal{S}_{bal}$. Thus, the spectrum of { $P_{ii}$} has a simple eigenvalue $1$ and the other eigenvalues have magnitudes strictly less than one \cite{xia2015structural}.
            
            \item When $S_i\in \mathcal{S}_{unbal}$, the opinion leaders corresponding to the sink $S_i$ have antagonistic interactions as well and form a strongly connected and structurally unbalanced subgraph. Then, it has been shown in \cite{xia2015structural} that all the eigenvalues of { $P_{ii}$} have magnitudes strictly less than one. 
        \end{itemize}

Let us consider the case when a sink $S_i$ is associated with one or more opinion leaders who are stubborn. Note that $S_i$ can belong to either $\mathcal{S}_{cp}$, $\mathcal{S}_{bal}$ or $\mathcal{S}_{unbal}$. In all of the cases, we know that if an opinion leader $j \in S_i$ is not stubborn, then $\sum_{k \in S_i}|p_{jk}|=1$. Obviously, $\sum_{k \in S_i}|p_{jk}|<1$ when $j$ is stubborn. Thus, the matrix $\tilde{P}_{(i+1)(i+1)}$ is row substochastic. Since the opinion leaders in sink $S_i$ are strongly connected, the spectral radius of $\rho(\tilde{P}_{(i+1)(i+1)})<1$ \cite{FB-LNS}. We deduce from Lemma \ref{lm:2} that $\rho(P_{(i+1)(i+1)})<1$. Thus, $P_{(i+1)(i+1)}$ associated with sink $S_i$ is convergent if one or more opinion leaders belonging in sink $S_i$ in $C(\mathcal{G})$ are stubborn.  

The preceding discussions imply that $P$ is semi-convergent { and not convergent} if and only if the network $\mathcal{G}$ possesses at least one opinion leader of type $\mathcal{V}_{o1}$ or a group of opinion leaders of type $\mathcal{V}_{o2}$ belonging to one or more sinks $S_i\in (\mathcal{S}_{bal}\cup \mathcal{S}_{cp})$ such that none of them is stubborn. Thus, it suffices to have a non-empty $\mathcal{S}_n$.
\end{proof}

Theorem \ref{lm:6} shows that $P$ is convergent when $\mathcal{S}_n=\emptyset$. Since $\rho(P)<1$ when $P$ is convergent, the steady state opinions are determined by taking the limit as $k \to \infty$ in eqn. \eqref{eq:opinion_dynamics} which results in:  
\begin{align}
\label{eqn:z_s_Sn_0}
 {\mathbf{z}}=(I-P)^{-1}\beta \mathbf{x}(0)   
\end{align}
where ${\mathbf{z}}=[z_1,...,z_n]=\lim_{k \to \infty}\mathbf{x}(k)$. The term $\beta \mathbf{x}(0)$ results in a vector with zero entries pertaining to the non-stubborn agents. Then, it follows from eqn. \eqref{eqn:z_s_Sn_0} that the final opinion of every agent in the network depends only on the initial opinions of the stubborn agents. 

{  On the other hand, if $\mathcal{S}_n \neq \emptyset$, then the opinions of the agents eventually converge to ${\mathbf{z}}$ given by,}
\begin{align}
\label{eqn:opinion_evol_With_stubbornness}
 {\mathbf{z}}=\lim_{k \to \infty}\bigg(P^k+\bigg(\sum_{i=0}^{k-1}P^i\bigg)\beta\bigg)\mathbf{x}(0).  
\end{align}

The final opinion vector to which the agents converge consists of two terms such that $\mathbf{z}=\mathbf{z_o}+\mathbf{z_s}$. Then,
\begin{itemize}
    \item {   $\mathbf{z_o}=\lim_{k \to \infty}P^k\mathbf{x}(0)$. Since $P$ is semi-convergent (and not convergent), $\mathbf{z}_o$ converges to,
\begin{equation}
    \label{eq:spanning_con}
        \mathbf{z}_o=\bigg(\sum_{j=1}^{|\mathcal{S}_n|} \mathbf{v}_j \mathbf{w}_j^T \bigg) \mathbf{x}(0)
    \end{equation}
    where  $\mathbf{w}_j$ and $\mathbf{v}_j$ are the left and right eigenvectors, respectively, of $P$ corresponding to $j^{th}$ eigenvalue $1$, such that $\mathbf{v}_{j}^T\mathbf{w}_j=1$. The left eigenvector $\mathbf{w}_j$ can be selected such that the nonzero entries of $\mathbf{w}_j$ correspond to the index of opinion leaders in the $j^{th}$ sink in $\mathcal{S}_n$ \cite{FB-LNS}. This implies that $\mathbf{z}_o$ takes into account the effect of the initial opinions of the opinion leaders that form the sinks in set $\mathcal{S}_n$ and is independent of the initial opinions of the stubborn agents. Hence,} we refer to it as the \textit{zero stubbornness response}.
    \item  $\mathbf{z_s}=\lim_{k \to \infty}(\sum_{i=0}^{k-1}P^i)\beta \mathbf{x}(0))$ takes into account the effect of initial opinions of the stubborn agents. We refer to it as the \textit{stubborn response}.
\end{itemize} 
  
\begin{remark}
\label{rem:4}
    Consider a network $\mathcal{G}$ devoid of any stubborn agents and structurally unbalanced sinks in $C(\mathcal{G})$. In this scenario, the opinions of agents in the entire network are governed by the cumulative effects of \textit{all} the opinion leaders in the network. {  Now, suppose we introduce stubbornness in the opinion leaders, then it follows from Theorem \ref{lm:6} and eqn. \eqref{eqn:z_s_Sn_0} that within a sink only the stubborn opinion leader(s) remain influential, while the other opinion leaders in the sink lose their influence.  } 
\end{remark} 

\begin{remark}
{ Theorem \ref{lm:6} highlights the interplay of network topology, signed interaction and stubbornness on influence of agents. The opinion leaders that form a group lose their influence if the interactions within the group do not follow Heider's Laws. It was shown in \cite{altafini,xia2015structural} that such agents converge to the `zero' opinion implying they are unable to make up their minds on the topic of decision. Studies on leadership such as \cite{BERNHEIM2020146} reveal a general preference of the electorate for decisive leaders. Therefore, the indecisive opinion leaders lose their influence.

Yet another scenario is when opinion leaders in a group lose their influence is if one or more opinion leaders in the group become stubborn. For instance, \textit{stubbornness} can be equivalent to the \textit{veto power} provided to the permanent members of the United Nations Security Council\cite{US_hegemony}. Once \textit{stubbornness} (or \textit{veto power}) is enforced, the opinions of the other members of the well-knit group become inconsequential. 

}\end{remark}
    
In this section, we discussed the role of the nature of interactions and stubbornness in deciding the influential agents in the network. Next, we quantify the influence each of these agents exert over the others in the network.
\section{Characterisation of influence using SFG}
\label{Sec5}
{ It follows from eqn. \eqref{eqn:opinion_evol_With_stubbornness} that the final opinions of agents depend on the initial states of the influential agents and the matrices $P$ and $\beta$. Therefore, the network interconnections and the stubbornness determine not only the influential agents but also the measure of their influence. In this work, we present a generalised framework to analyse the effects of network interconnections, their signs and stubbornness in a weakly connected signed network on the opinion formation process.}

Consider a network $\mathcal{G}$ with stubborn agents where the opinions of agents are governed by eqn. \eqref{eq:opinion_dynamics} such that $P$ is semi-convergent (and not convergent). By taking the limit as $k \to \infty$ in eqn. \eqref{vector_op_model}, it follows that,
\begin{equation}
\label{eqn:steady_state_OD}
 \mathbf{z}=P\mathbf{z}+\beta \mathbf{x}(0)   
\end{equation}
{  
Let $\mathbf{y}\in \mathbb{R}^{n+s}$ be defined as $\mathbf{y}=[y_1,...,y_{n+s}]=[z_1,...,z_n,x_{k_1}(0),...,x_{k_s}(0)]$ where $x_{k_h}(0)$ denotes the initial opinion of the stubborn agent $k_h$, $h\in\{1,2,..,s\}$ and $s$ is the number of stubborn agents in $\mathcal{G}$. Then, eqn. \eqref{eqn:steady_state_OD} can be re-arranged as,
\begin{align}
\label{eqn:y}
    \mathbf{y}=B\mathbf{y}
\end{align}
where matrix $B=[b_{ij}]$ is defined as,
\begin{align}
\label{eqn:matrix_B}
   B=\begin{bmatrix}
 P & \widetilde{\beta} \\
 0 & I_s
\end{bmatrix} 
\end{align}
with $\widetilde{\beta}\in \mathbb{R}^{n,s}$ which can be derived from eqn. \eqref{eqn:steady_state_OD}.}
\subsection{Construction of SFGs}
A set of linear equations of the form given in \eqref{eqn:y} can be represented by an \textit{SFG} $\mathcal{G}_s={ ({V}_s,B_s)}$ \cite{mason1956feedback}. In our framework, the nodes $\{1,2,...,n+s\}$ of $\mathcal{G}_s$ are associated with the node signals { $\{y_1,y_2,...,y_{n+s}\}$ corresponding to the final opinions of the $n$ agents and the initial opinions of $s$ stubborn agents. For any two nodes $p$ and $q$ in $\mathcal{G}_s$, the branch gain $g_{p,q}=b_{pq}$.} Note that the direction of the branches of $\mathcal{G}_s$ is the reverse of the flow of information in the network $\mathcal{G}$. 
The sources and sinks in $\mathcal{G}_s$ are nodes having only outgoing and only incoming branches, respectively. A non-source and a non-sink node in an SFG are nodes which are not source and sinks, respectively.

\begin{expm}
\label{expm:1}
 Consider the network $\mathcal{G}=(\mathcal{V},\mathcal{E},A)$ shown in {  Fig. \ref{fig:Network_with_cycles}} with edge weight equal to the entries of adjacency matrix. As defined in Sec \ref{subsec:COA}, the nodes $\{1,...,4\}$ are followers, and $\{5,...,$ $11\}$ are opinion leaders associated with the sinks $S_1,$ $S_2$ and $S_3$ of $C(\mathcal{G})$ as highlighted in Fig. \ref{fig:Network_with_cycles}. Consider agents $1$ and $6$ to be stubborn. The effectively structurally balanced sinks of $C(\mathcal{G})$ include: $\mathcal{S}_b=\{S_1,S_2,S_3\}$ and $\mathcal{S}_n=\{S_1,S_3\}$. The opinions of agents in the network evolve according to eqn. \eqref{eq:opinion_dynamics} for the initial opinions $\mathbf{x}(0)=[8.0,9.0,$ $6.0,5.0,7.0,$ $3.0,6.5,-10.0,7.0,0.3,2.5]$ with $\beta_1=0.3$ and $\beta_6=0.2$. It follows from eqn. \eqref{eqn:y}, that the SFG $\mathcal{G}_s$ derived from $\mathcal{G}$ has $13$ nodes associated with states $\{z_1,z_2,...,z_{11},x_1(0),x_6(0)\}$ with branch gains $g_{i,j}=b_{ij}$. The nodes $5,12$ and $13$ of the SFG form sources, while the rest form non-source nodes. Note that opinion leaders $9,10$ and $11$ are influential according to Theorem \ref{lm:6}. However, the nodes corresponding to their final opinions do not form sources in $\mathcal{G}_s$.
 
\begin{figure}[h]
    \centering
    \includegraphics[width=0.45\textwidth]{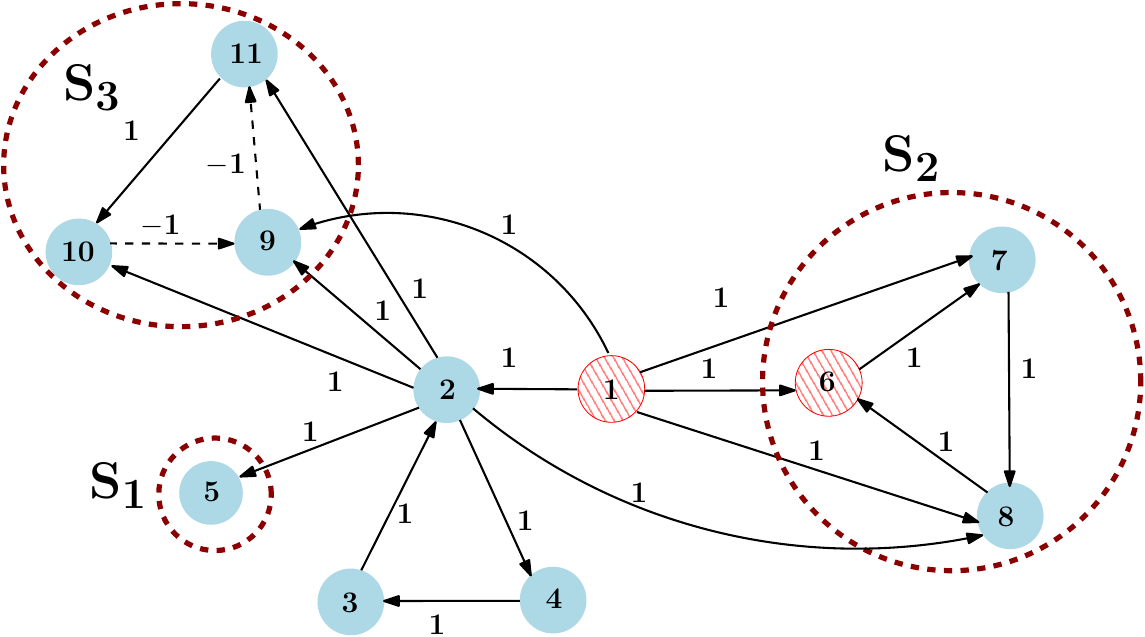}
    \caption{A weakly connected digraph with stubborn agents $1$ and $6$.}
    \label{fig:Network_with_cycles}
   \end{figure}
\end{expm}

\begin{remark}
\label{rem:SFG_Red}
{ An SFG is used in a multi-input, multi-output (MIMO) system to determine the effect of inputs (or sources) on different outputs (or sinks). The gain of an SFG for a given source and  sink pair is evaluated using eqn. \eqref{eq:sfg-gain}. 
In our framework, even groups of interacting opinion leaders can drive the opinions of all the agents in the network. In the constructed SFG $\mathcal{G}_s$, their final opinions do not form sources.  
Thus, in order to examine the influence of all such agents, we can construct reduced SFGs as discussed below. 
}
\end{remark}

In a group of $n$ agents connected over a weakly connected digraph $\mathcal{G}$ whose opinions evolve according to eqn. \eqref{eq:opinion_dynamics}, the SFG $\mathcal{G}_s=(V_s, B_s)$ derived from eqn. \eqref{eqn:y} is reduced to $\bar{\mathcal{G}}_s=(\bar V_s,\bar B_s)$ by the following steps.
We begin with the construction of node set $\bar{V}_s$ from $V_s$,
\begin{itemize}
    \item A node $i$ in $ \mathcal{G}_s$, which is a follower in $\mathcal{G}$ for $i\in\{1,...,$ $m\}$, forms a non-source node indexed $\bar i$ in $\bar{\mathcal{G}}_s$ with state $\bar{y}_i=y_i$.
    \item A node $i$ in $\mathcal{G}_s$, which is an opinion leader in $\mathcal{G}$ for $i\in\{m+1,...,n\}$, 
    \begin{itemize}
        \item if $i \in \mathcal{V}_{o1}$ in $\mathcal{G}$, it forms a source $\mathcal{O}_r$ with state $\bar{y}_{\mathcal{O}_r}=x_i(0)$.
        \item if $i \in \mathcal{V}_{o2}$ in $\mathcal{G}$ is associated with sink $S_l$ of $C(\mathcal{G})$ such that $S_l \in \mathcal{S}_{cp} \cap \mathcal{S}_n$, then all such $i$ with $\mathcal{S}^v(i)=S_l$ collectively form a source $\mathcal{O}_r$ in $\bar{\mathcal{G}}_s$ with state $\bar{y}_{\mathcal{O}_r}=(\mathbf{w}_{P}^{l+1})^T \mathbf{x}_{l+1}(0)$. 
        \item if $i \in \mathcal{V}_{o2}$ in $\mathcal{G}$ is associated with sink $S_l$ of $C(\mathcal{G})$ such that $S_l \in \mathcal{S}_{bal} \cap \mathcal{S}_n$, then all such $i$ with $\mathcal{S}^v(i)=S_l$ form two sources $\mathcal{O}_r$ and $\mathcal{O}_{r+1}$ in $\bar{\mathcal{G}}_s$ corresponding to the bipartition of $S_l$. The sources are associated with states $\bar{y}_{\mathcal{O}_r}=a$ and $\bar{y}_{\mathcal{O}_{r+1}}=-a$ where $a=(\mathbf{w}_{P}^{l+1})^T \mathbf{x}_{l+1}(0)$.
         \item if $i\in \mathcal{V}_{o2}$ in $\mathcal{G}$ is associated with sink $S_l$ of $C(\mathcal{G})$ such that either $i$ is stubborn or there exists a stubborn $h \in \mathcal{V}_{o2}$ such that $\mathcal{S}^v(h)=S_l$, then node $i$ forms a non-source node $\bar i$ with state $\bar{y}_i=y_i$.
          \item if $i \in \mathcal{V}_{o2}$ in $\mathcal{G}$ is associated with sink $S_l$ of $C(\mathcal{G})$ such that $S_l \in \mathcal{S}_{unbal}$ and none of the opinion leaders in $S_l$ is stubborn, then $i$ does not form a node in $\bar{\mathcal{G}}_s$.
        \end{itemize}
        \item Each node $i \in \mathcal{G}_s$ for $i\in\{n+1,...,n+s\}$ forms a source $\mathcal{O}_r$ in $\bar{\mathcal{G}}_s$ with state $\bar{y}_{\mathcal{O}_r}=y_i$.
    \end{itemize}
    where $l\in \{1,...,n_s\}$, $r\in\{1,...,|\mathcal{V}_{o1}|+s+|\mathcal{S}_{cp}|+2|\mathcal{S}_{bal}|-1\}$, $\mathbf{w}_{P}^{(l+1)}$ is the normalised left eigenvector corresponding to simple eigenvalue $1$ of the $(l+1)^{th}$ block of $P$ denoted by $P_{(l+1)(l+1)}$ and $\mathbf{x}_{l+1}(0)$ is the vector composed of initial opinions of opinion leaders in sink $S_l$.
    The following are the steps for determining the branches in $\bar B_s$ and their respective branch gains that exist between any two nodes in $\bar{V}_s$:
 \begin{itemize}
\item for any non-source node $\bar i \in \bar{V}_s$ and source $\mathcal{O}_r$: 
\begin{itemize}
    \item if $\mathcal{O}_r$ corresponds to an opinion leader $j$ in $\mathcal{V}_{o1}$, the branch gain of branch $(\mathcal{O}_r,\bar i)$ is $g_{\bar i,\mathcal{O}_r}=b_{ij}$.
    \item if $\mathcal{O}_r$ corresponds to opinion leaders in a sink $S_l \in \mathcal{S}_{cp}$, the branch gain of branch $(\mathcal{O}_r,\bar i)$ is $g_{\bar i,\mathcal{O}_r}=\sum_{h \in S_l}b_{ih}$.
    \item if $\mathcal{O}_r$ corresponds opinion leaders in a partition $\mathcal{V}_{lq}$ of sink $S_l \in \mathcal{S}_{bal}$ for $q\in\{1,2\}$, the branch gain of branch $(\mathcal{O}_r,\bar i)$ is $g_{\bar i,\mathcal{O}_r}=\sum_{h \in \mathcal{V}_{lq}}b_{ih}$.
    \item if $i$ is stubborn agent in $\mathcal{G}$ and $\mathcal{O}_r$ corresponds to its initial opinion, then the branch gain of branch $(\mathcal{O}_r,\bar i)$ is $g_{i,\mathcal{O}_r}=\beta_i$
\end{itemize}
    \item for any pair of non-source nodes $\bar i,\bar j \in \bar{V}_s$, the branch gain of branch $(\bar j, \bar i)$ is $g_{\bar i,\bar j}=b_{ij}$.
    \end{itemize}        
The derivation of the rules for constructing SFGs is given in Appendix. We illustrate the construction of the SFG $\bar{\mathcal{G}}_s$ using the following example.
\begin{expm}
\label{expm:3}
Consider the network $\mathcal{G}$ in Fig. \ref{fig:Network_with_cycles} whose SFG $\mathcal{G}_s$ was constructed in Example \ref{expm:1}. Now, we derive $\bar{\mathcal{G}}_s$. 
The node $i\in \{1,...,4\}$ in $\mathcal{G}_s$ corresponds to a follower so it forms a non-source node $\bar{i}$ in $\bar{\mathcal{G}}_s$. 
Node $5$ corresponds to opinion leader in $\mathcal{V}_{o1}$ so it forms source $\mathcal{O}_1$. 
The nodes $i\in\{6,7,8\}$ in $\mathcal{G}_s$ correspond to the opinion leaders in sink $S_2 \in \mathcal{S}_{cp}$ but since $6$ is stubborn they form non-source nodes $\bar i$. 
The nodes $i\in\{9,10,11\}$ correspond to opinion leaders that associated with sink $S_3\in \mathcal{S}_{bal}$ with $9\in \mathcal{V}_{31}$ and $10,11\in \mathcal{V}_{32}$. They form sources $\mathcal{O}_2$ and $\mathcal{O}_3$ respectively. Finally, the nodes $12$ and $13$ corresponding to the initial opinions of stubborn agents $1$ and $6$ form sources $\mathcal{O}_4$ and $\mathcal{O}_5$, respectively. 
The states corresponding to the nodes in $\bar{\mathcal{G}}_s$ are derived using the rules discussed above.

The branch gain of branch $(\bar j,\bar i)$ if $\bar i \in \{\bar 1,..,\bar 4\}$ and $\bar j=\mathcal{O}_1$ is $g_{\bar i,\mathcal{O}_1}=p_{i5}$. 
For $j\in \{\mathcal{O}_2,\mathcal{O}_3\}$, the branch gain $g_{\bar i,\mathcal{O}_2}=\sum_{h \in \mathcal{V}_{31}}p_{ih}$ and $g_{\bar i ,\mathcal{O}_3}=\sum_{h \in \mathcal{V}_{32}}p_{ih}$. The sources $\mathcal{O}_4$ and $\mathcal{O}_5$ 
have branches in $\bar{\mathcal{G}}_s$ with branch gain $g_{\bar 1,\mathcal{O}_4}=\beta_1$ and $g_{\bar 6,\mathcal{O}_5}=\beta_6$. For the branches between non-source nodes, $g_{\bar i,\bar j}=b_{ij}=p_{ij}$. { The reduced SFG $\bar{\mathcal{G}}_s$ constructed from $\mathcal{G}_s$ is given in Fig. \ref{FIG:Reduced_SFG}. To avoid confusion, the effect of each source is shown in the sub-figures \ref{Source_1_follower_1}-\ref{Source_5_follower_1} independent of the other sources (source $\mathcal{O}_4$'s SFG is omitted due to its simplicity). Further, in each sub-figure corresponding to a source only those non-source nodes are represented to which the source has a forward path (only these nodes will be affected by the source \cite{mason1956feedback}).  
}
\end{expm}
\begin{figure}    
\centering
\begin{subfigure}{1\linewidth}
       \centering
        \includegraphics[width=1\textwidth, height=2.5cm,keepaspectratio]{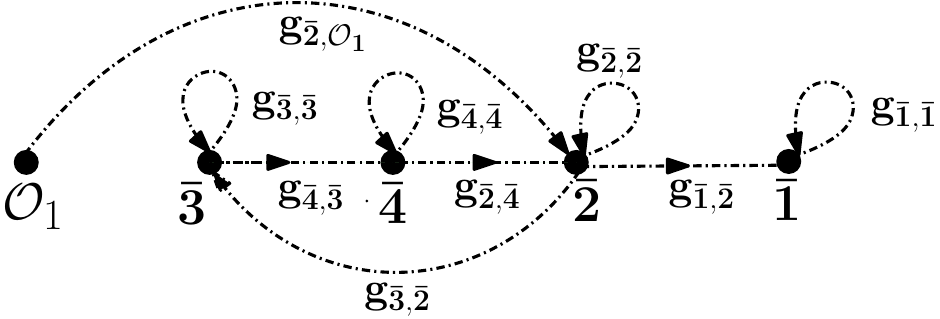}
        \caption{SFG with Source $\mathcal{O}_1$}
        \label{Source_1_follower_1}
 \end{subfigure}
    \begin{subfigure}{1\linewidth}
        \centering
        \includegraphics[width=1\textwidth, height=2.7cm,keepaspectratio]{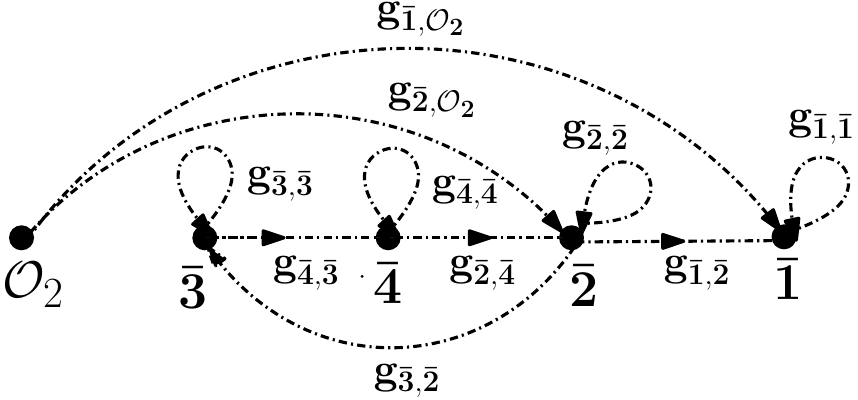}
        \caption{SFG with source $\mathcal{O}_2$}
        \label{Source_2_follower_1}
    \end{subfigure}
    \begin{subfigure}{0.45\textwidth}
        \centering
        \includegraphics[width=0.8\textwidth]{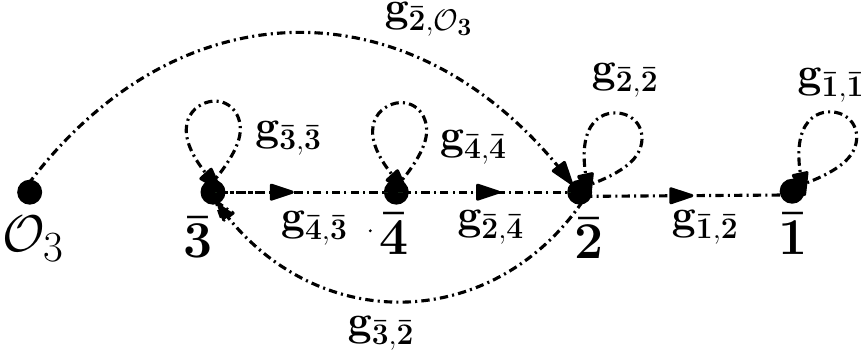}
        \caption{SFG with source $\mathcal{O}_3$}
        \label{Source_4_follower_1}
    \end{subfigure}
    \begin{subfigure}{0.5\textwidth}
        \centering
        \includegraphics[width=1\textwidth,height=2.7cm,keepaspectratio]{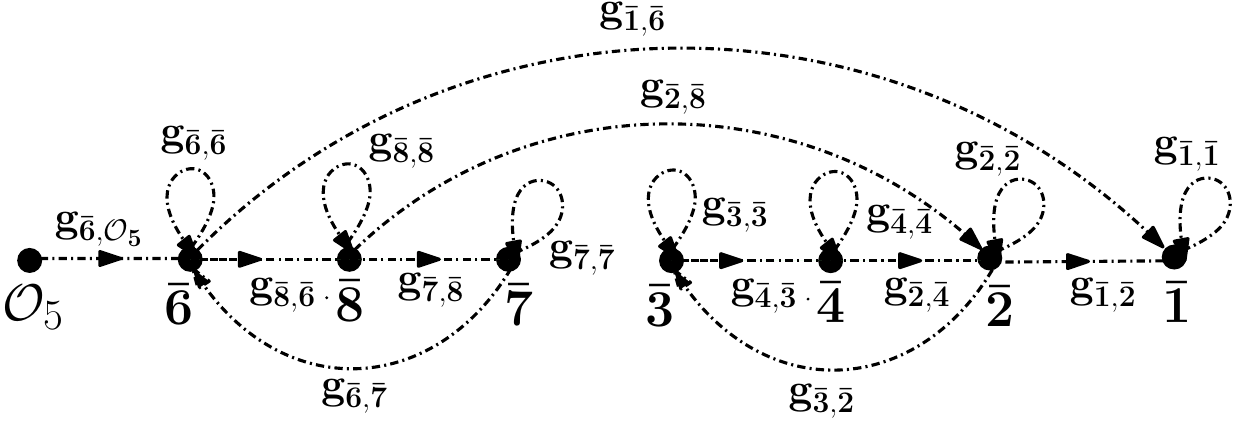}
        \caption{SFG with source $\mathcal{O}_5$}
        \label{Source_5_follower_1}
    \end{subfigure}
    \caption{The SFGs show the effects of each source on the rest of the nodes. Source $\mathcal{O}_4$ is excluded.}
    \label{FIG:Reduced_SFG}
\end{figure}

The above-mentioned procedure allows us to construct the reduced SFG $\bar{\mathcal{G}}_s$ from $\mathcal{G}_s$. { Each source is associated with an influential agent(s) whether stubborn or non-stubborn.} Next, we use the reduced SFG $\bar{\mathcal{G}}_s$ to determine the effect of sources constituted of influential agents on the rest of the nodes. 

\subsection{Quantification of Influence}
{ Like $\mathcal{G}_s$, the reduced SFG $\bar{\mathcal{G}}_s$ also represents a set of linear equations, so} it satisfies the \textit{superposition principle}. As a result, the effect of a source $\mathcal{O}_r$ on a node $\bar i$ in $\bar{\mathcal{G}}_s$
can be determined by considering the effect of the other sources to be zero as illustrated in \cite{mason1956feedback} where {  $\bar i\in\bar{V}_s$ and $r\in\{1,...,|\mathcal{V}_{o1}|+s+|\mathcal{S}_{cp}|+2|\mathcal{S}_{bal}|\}$. We know that the gain of the SFG is the signal obtained at a sink for per unit signal at the source applied at the source. If the effect of sources is to be determined at node $\bar i$ that has outgoing branches, it is not a sink in $\bar{\mathcal{G}}_s$. To determine the effect of source on $\bar i$, we link a new node $\delta_{i}$ to $\bar i$ via a branch $(\bar i,\delta_{i})$ with branch gain $g_{\delta_{i},\bar{i}}=1$. The node $\delta_{i}$ has no outgoing and only one incoming branch. Thus, it forms the sink $\delta_{i}$ and its node state value $\bar{y}_{\delta_{i}}$ is equal to $\bar{y}_i$.} 
The collective influence coefficient $c_{ir}$ is the gain of the SFG given by eqn. \eqref{eq:sfg-gain} when only source $\mathcal{O}_r$ is acting and $\bar i$ or $\delta_{i}$ is the sink. 

When the source $\mathcal{O}_r$ is composed of opinion leaders in $\mathcal{V}_{o2}$, the collective influence coefficient $c_{ir}$ gives an overall effect of all the opinion leaders in source $\mathcal{O}_r$ on $i$. It does not specify the individual influence of each opinion leader comprising the source. Thus, we define the individual influence coefficient $\theta_{ij}$ to measure the exact contribution of an influential agent $j$ on the final opinion of node $i$ which is determined in the following result.

\begin{thm}
\label{Thm:4}
Consider a group of $n$ agents connected over a weakly connected digraph $\mathcal{G}$. Under the opinion dynamics model given in eqn. \eqref{eq:opinion_dynamics}, the opinion of an agent $i\in \mathcal{V}$ eventually converges to 
$z_i=\sum_{j \in \mathcal{V} } \theta_{ij}x_j(0)$ where $\theta_{ij}$, the individual {  influence coefficient, is defined as},
\begin{align}
\label{eqn:IIC}
    \theta_{ij}=\begin{cases}
        c_{ir} & j \in \mathcal{V}_{o1}{  \cup \mathcal{V}_s} \\
        c_{ir}(\mathbf{w}_P^{l+1})_{\kappa}^T & j \in \mathcal{V}_{o2}, S_l \in \mathcal{S}_{cp} { \cap \mathcal{S}_n}\\
        (\sigma_{1}c_{ir}+\sigma_2c_{i(r+1)})(\mathbf{w}_{P}^{l+1})_{\kappa} & j \in \mathcal{V}_{o2}, S_l \in S_{bal}{ \cap \mathcal{S}_n} \\
        (\mathbf{w}_{P}^{l+1})_{\kappa} & i,j \in S_l \in \mathcal{S}_{cp} \cap \mathcal{S}_n \\
      \sigma_i (\mathbf{w}_{P}^{l+1})_{\kappa} & i,j \in S_l \in \mathcal{S}_{bal} \cap \mathcal{S}_n \\
        0 & \text{otherwise}
    \end{cases}
\end{align}
where $c_{ir}$ is the collective influence coefficient for source $\mathcal{O}_r$ and sink corresponding to node $\bar i$ in $\bar{\mathcal{G}}_s$, $\mathcal{O}_r$ is the source in $\bar{\mathcal{G}}_s$ associated with $j$, $r\in\{1,...,|\mathcal{V}_{o1}|+s+|\mathcal{S}_{cp}|+2|\mathcal{S}_{bal}|\}$, $S_l=\mathcal{S}^v(j)$, $\kappa=j-(m+\sum_{g=1}^{l-1}|S_g|)$,  and $\sigma_i\in\{1,-1\}$ depending on the partition of $S_l$ to which opinion leader $i$ belongs. 
\end{thm}


Theorem \ref{Thm:4} quantifies the influence of each influential agent on every agent in the network. The preceding analysis establishes a relation between the network topology (the paths and cycles gain) and the influence of an agent by the use of SFG in the determination of influence. 
We illustrate these results through the following example.
\begin{expm}
\label{expm:4}
{   For the given network in Fig. \ref{fig:Network_with_cycles}, the reduced SFG $\bar{\mathcal{G}}_s$ was constructed in Example \ref{expm:3}. In this example, we determine the gain of the reduced SFG $\bar{\mathcal{G}}_s$ to quantify the influence distribution. 
}
Let us quantify the influence exerted by the influential agents $\{5,6,8,9,10\}$ in $\mathcal{G}$ on agent $1$.
{ If node $\bar 1$ in $\bar{\mathcal{G}}_s$ has a self-loop, then a sink ${\delta_{1}}$ is added such that $g_{\delta_{1},\bar 1}=1$.} For each source $\mathcal{O}_r$ associated with influential agent(s) and sink ${\delta_{1}}$, coefficient { $c_{1h}$} is determined from eqn. \eqref{eq:sfg-gain} in terms of branch-weights. The gain of the SFG for sources $\mathcal{O}_1$ and $\mathcal{O}_2$ and sink ${\delta_{1}}$ is given as follows:
\begin{subequations}
\begin{align}
\begin{split}
 c_{ 1 1}&=\frac{g_{\bar 2, \mathcal{O}_1}g_{ \bar 1,\bar 2}( 1-g_{\bar 3,\bar 3}-g_{\bar 4,\bar 4}+g_{\bar 3,\bar 3}g_{\bar 4,\bar 4})}{\Delta_{ 1}}  
\end{split}\\
\begin{split}
c_{ 1 2}&=\frac{g_{\bar 2,\mathcal{O}_2}g_{ \bar 1,\bar 2}( 1-g_{\bar 3,\bar 3}-g_{\bar 4,\bar 4}+g_{\bar 3,\bar 3}g_{\bar 4,\bar 4})+g_{\bar  1,\mathcal{O}_2}\Delta_2}{\Delta_1}
\end{split}
\end{align}
\label{eqn:CII}
\end{subequations}
where $\Delta_{1}=\bar 1-(g_{\bar 1,\bar 1}+g_{\bar 2,\bar 2}+g_{\bar 3,\bar 3}+g_{\bar 4,\bar 4}+g_{\bar 4,\bar 3}g_{\bar 3,\bar 2}g_{\bar 2,\bar 4})+$ 
 $(g_{\bar 1,\bar 1}g_{\bar 2,\bar 2}+g_{\bar 2,\bar 2}g_{\bar 3,\bar 3}+g_{\bar 1,\bar 1}g_{\bar 4,\bar 4}+g_{\bar 1,\bar 1}g_{\bar 3,\bar 3}+g_{\bar 4,\bar 4}g_{\bar 2,\bar 2}+g_{\bar 3,\bar 3}g_{\bar 4,\bar 4}+g_{\bar 1,\bar 1}g_{\bar 4,\bar 3}g_{\bar 3,\bar 2}g_{\bar 2,\bar 4})-(g_{\bar 1,\bar 1}g_{\bar 2,\bar 2}g_{\bar 3,\bar 3}+g_{\bar 1,\bar 1}g_{\bar 2,\bar 2}g_{\bar 4,\bar 4}+g_{\bar 1,\bar 1}g_{\bar 4,\bar 4}g_{\bar 3,\bar 3}+g_{\bar 4,\bar 4}g_{\bar 2,\bar 2}g_{\bar 3,\bar 3})+g_{\bar 1,\bar 1}g_{\bar 2,\bar 2}g_{\bar 3,\bar 3}g_{\bar 4\bar 4}$ and 
$\Delta_2= 1-g_{\bar 4,\bar 4}-g_{\bar 3,\bar 3}-g_{\bar 2,\bar 2}-g_{\bar 4,\bar 3}g_{\bar 2,\bar 4}g_{\bar 3,\bar 2}+g_{\bar 4,\bar 4}g_{\bar 3,\bar 3}+g_{\bar 2,\bar 2}g_{\bar 4,\bar 4}+g_{\bar 3,\bar 3}g_{\bar 2,\bar 2}-g_{\bar 2,\bar 2}g_{\bar 3,\bar 3}g_{\bar 4,\bar 4}$.
Similarly, the coefficients $c_{ir}$ are derived for each pair of source and sink (corresponding to each non-source node) in $\bar{\mathcal{G}}_s$.
The values of the collective influence coefficients $c_{ir}$ are given in Table \ref{table:1}.
\begin{table}[h!]
\centering
\begin{tabular}{||c| c c c c c c c||} 
 \hline
 Source & 1 & 2 & 3 & 4 & 6 & 7 & 8\\ [0.5ex] 
 \hline\hline
$\mathcal{O}_1$ & 0.02 & 0.2 & 0.2 & 0.2 & 0 &0 &0 \\ 
$\mathcal{O}_2$  & 0.12 & 0.2 & 0.2 & 0.2& 0 &0 &0  \\
$\mathcal{O}_3$  & 0.04 & 0.4 & 0.4 &0.4& 0 &0 &0 \\
$\mathcal{O}_4$& 0.5 & 0 & 0 &0& 0 &0 &0\\ 
$\mathcal{O}_5$& 0.32 & 0.2 & 0.2 &0.2& 1 &1 &1\\ [1ex] 
 \hline
\end{tabular}
\caption{Collective influence coefficient $c_{ir}$ for source $\mathcal{O}_r$ and node $i$ in $\bar{\mathcal{G}}_s$ }
\label{table:1}
\end{table}

Next, we evaluate the individual influence coefficient $\theta_{ij}$ for $i,j \in \mathcal{V}$ using the collective influence coefficients in Table \ref{table:1}.
For $i\in\{1,...,4,6,7,8\}$, influence of stubborn agent $1$ is $\theta_{i1}=c_{i4}$, influence of opinion leader $5$ is $\theta_{i5}=c_{i1}$ and influence of stubborn agent $6$ is $\theta_{i6}=c_{i5}$.
Since sink $S_3$ is structurally balanced, thus, the individual effect of an opinion leader $j \in \{9,10,11\}$ is $\theta_{ij}=(c_{i2}-c_{i3})(w_P^{4})_\kappa, \kappa\in \{1,2,3\}$ where $w_P^{4}=[0.2941, -0.3137,-0.3922]$. For opinion leaders $i,j \in S_3$, the coefficient $\theta_{ij}=\sigma_{i}(\mathbf{w}_P^4)_{\kappa}$, where $\sigma_9=1,\sigma_{10}=\sigma_{11}=-1$.
We calculate the final opinion ${z}_1$ of follower $1$ equal to ${z}_1=\sum_{j \in \mathcal{V}}{\theta}_{1j}x_{j}(0)=0.5 \times8 + 0.02\times 7+0.02\times 7-0.02\times0.3 -0.03 \times 2.5  +0.32 \times 3 =5.15$ as obtained in numerical simulation in Fig. \ref{fig:sim_2}. Similarly, it holds for the other agents as well. 
Fig. \ref{fig:sim_2} shows that the opinions in $\mathcal{G}$ that are governed by opinion model \eqref{eq:opinion_dynamics} eventually converge.
\end{expm}
\begin{figure}[h]
    \centering
    \includegraphics[width=0.5\textwidth]{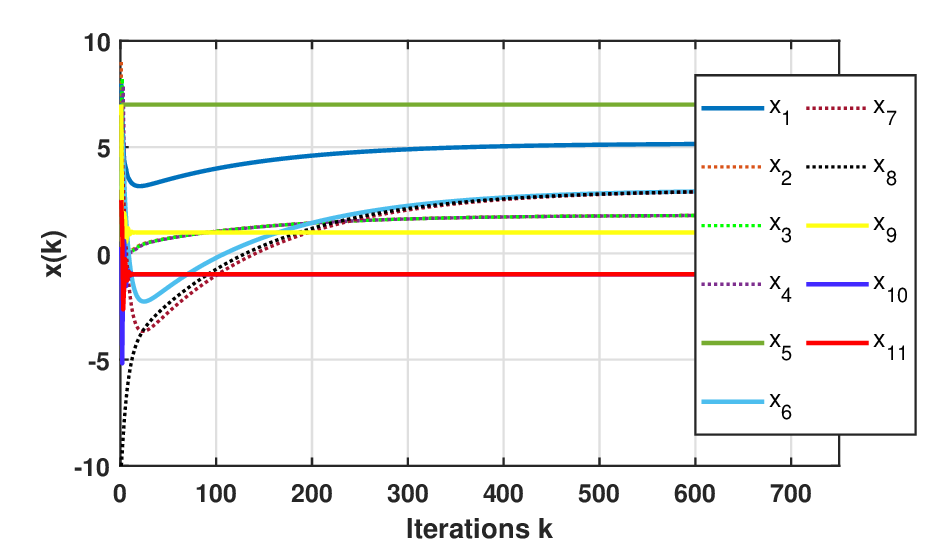}
    \caption{The evolution of opinions in the network given in Fig. \ref{fig:Network_with_cycles} governed by eqn. \eqref{eq:opinion_dynamics}. The opinions of influential agents are indicated by solid lines and the rest by dotted lines.}
    \label{fig:sim_2}
\end{figure}

 \begin{remark}
    The SFG obtained for a source $\mathcal{O}_r$ (composed of the opinion leaders of a sink $S_l$) is a graphical depiction of the flow of influence within the network. We know that the gain of the SFG depends on the magnitude of signed interactions along the directed paths. \textit{The influence of a source on a follower can be increased (decreased) directly by increasing (decreasing) the magnitudes of the branch gains along the forward path.} The SFG gain also depends on the cycles along the directed paths from the source to the follower. \textit{Hence, the influence of the source can also be increased by increasing the magnitudes of the loop gains of the positive cycles associated with the corresponding SFG.}
\end{remark} 

{ \begin{expm}
\label{expm:change}
    Suppose the communication among the agents change causing a change in the signs of the edges $(1,6)$ and $(2,10)$ in the network $\mathcal{G}$ shown in Fig. \ref{fig:Network_with_cycles}. The opinions at $k^{th}$ instance and final opinions after the change in the network are represented by $\hat{\mathbf{x}}(k)$ and $\hat{\mathbf{z}}$, respectively. Now, if we construct the reduced SFG $\bar{\mathcal{G}}_s$ from this new network, the sources and their associated states do not change. Only the branch gain of branches $g_{6,1}$ and $g_{\mathcal{O}_3,2}$ change.  Since the nodes $6,7$ and $8$ are affected only by $\mathcal{O}_5$ and the altered branch does not lie in the path, it follows that their opinions do not change.  Finally, the deviation in the opinions of nodes $1,...,4$ is presented in Fig. \ref{fig:sim_3}. The absolute deviation in the opinions of agents due to the change comes out to be  ${\sum_{i=1}^{n}(|z_i-\hat{z}_i|)/n}=0.15$ where $\hat{z}_i$ denote final opinion of $i$ after the alteration of signs in the network.

\begin{figure}[h]
    \centering
    \includegraphics[width=0.5\textwidth]{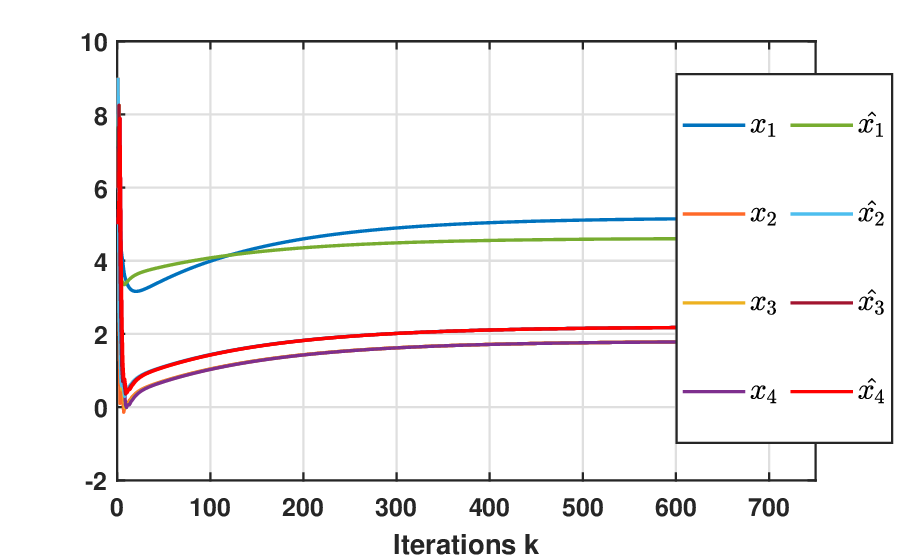}
    \caption{Effect of sign switching in network on agent's opinions.}
    \label{fig:sim_3}
\end{figure}
\end{expm}

\begin{remark}
In Example 5, we observe that even the simple modifications in network structures, which do not affect the nature of any influential node, can result in noticeable differences in the final opinion. 
This implies that both the initial opinions of influential agents and the network interactions along with their signs play an important role in determining the agents' influence. In real-world scenarios, when the objective is to modify the opinions in a network to achieve a reduction in polarisation, conflicts, \textit{etc}, changing an influential agent's opinion might not be feasible. For example, social media organisations try to control the interactions among individuals using recommender systems; effectively, they are modifying the underlying network structures \cite{racz2023towards}.
\end{remark}}
Till now, we have discussed the impact of the network topology, signed interactions and stubbornness in deciding the influence of influential agents on the final opinions of others in the network. Next, we quantify their effect on the degree of influence of the influential agents on the overall opinion formation among agents in the network. 
\section{Absolute Influence centrality}
\label{Sec6}
In this subsection, we propose a centrality measure to determine the overall contribution of an influential agent in the final opinions of agents in the network. We employ the influence coefficient derived in Theorem \ref{Thm:4} and construct a matrix $\Theta=[\theta_{ij}]$.  Since the entries of $\Theta$ account for the effects of the initial opinions of the influential agents on the final opinion vector $\mathbf z$ of the agents. 
Thus,
\begin{align}
\label{eq:inf_x0}
    \mathbf{z}=\Theta\mathbf{x}(0)
\end{align}
\begin{figure}[t]
    \centering
    \includegraphics[width=0.45\textwidth]{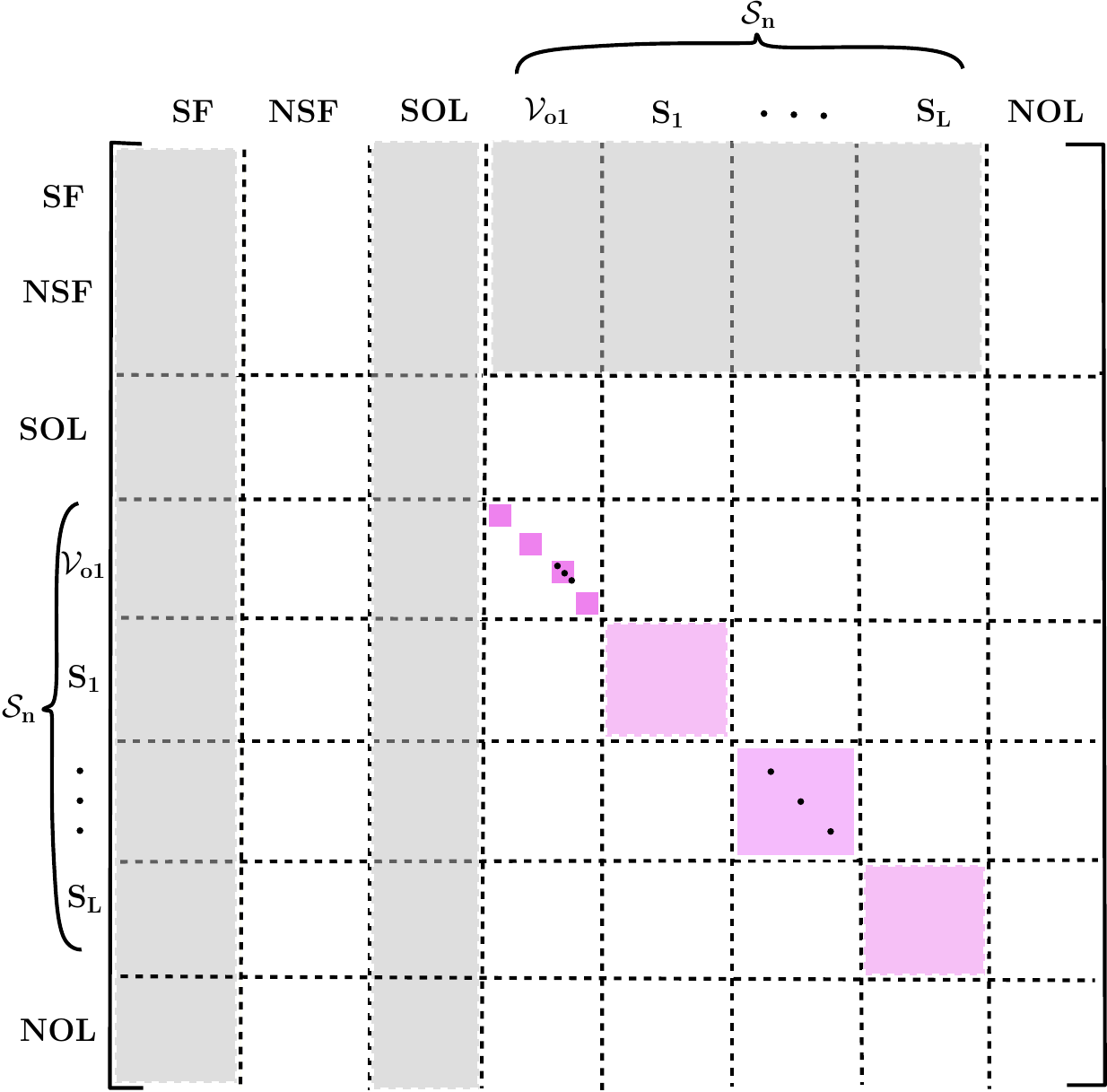}
    \caption{A representation of $\Theta$ to see who influences whom. In the figure, $SF$ stands for stubborn followers, $NSF$ for non-stubborn followers, $SOL$ for stubborn opinion leaders and $NOL$ for non-influential opinion leaders.}
    \label{fig:influence_matrix}
\end{figure}
Fig. \ref{fig:influence_matrix} depicts the non-zero pattern of influence matrix $\Theta$ for any network. 
In the grey region in the figure, $\theta_{ij}$ can be zero or non-zero depending on network connectivity. The entries in pink represent the regions where $\theta_{ij} \neq 0$ implying that $j$ definitely influences $i$. The entries in white signify zero influence of $j$ on $i$ \textit{i.e.} $\theta_{ij}=0$. Now, owing to the antagonism in the network, the influence of an agent may be positive or negative. Thus, we quantify the influence in the network using the matrix $\Tilde{\Theta}=[|\theta_{ij}|]$ where $|\theta_{ij}|$ is the absolute value of $(i,j)^{th}$ entry of $\Theta$.


%
{  An \textit{influential} agent dictates the final opinions of one or more agents in the network. Its degree of influence} is determined by the extent to which it alters the final opinion pattern of the network for a unit change in its initial opinion. For a given network $\mathcal{G}$ and the initial opinion pattern $\mathbf{x}(0)$, if the initial opinion $x_i(0)$ of an agent $i \in \mathcal{V}$ deviates to $x_i(0)+\delta$ resulting in a new pattern of initial opinion denoted by $\mathbf{x}^i(0)$. The change in final opinion is given by, 
\begin{align}
\label{infl_centrality}
   \Delta \mathbf{z}^i= \frac{\|\mathbf{z}^i-\mathbf{z}\|_1}{|\delta|}.
\end{align}
where $\mathbf{z}$ and $\mathbf{z}^i$ are the final opinion patterns for $\mathbf{x}(0)$ and $\mathbf{x}^i(0)$, respectively.

We define the absolute influence centrality of an agent $i$ as $\Delta \mathbf{z}^i$, which is the deviation caused by agent $i$ in the final opinion pattern. The most influential agent is the one that causes the maximum deviation. Thus, an agent $p$ is the most influential agent if,
{ \begin{align*}
p = \arg \max_{i \in \mathcal{V}}\Delta \mathbf{z}^i
\end{align*}}
The following result relates absolute influence centrality with the $\Theta$ matrix.
\begin{thm}
\label{thm:centrality}
    Consider a network of $n$ agents connected over a weakly connected graph whose opinions evolve by the proposed opinion model \eqref{eq:opinion_dynamics} in the presence of stubbornness. The measure of the influence centrality of all the agents in the network is given by the \textit{absolute influence centrality} $\Tilde{\Theta}^T \mathbb{1}_n$.
\end{thm}
\begin{proof}
In order to determine the influential nodes, we perturb the initial positions of the agents one after the other. For an initial opinion pattern $\mathbf{x}(0)$, suppose the initial opinion of the $i^{th}$ is perturbed by $\delta$, then we represent the new initial opinion pattern as $\mathbf{x}^i$. The antagonism in the network can cause the influence of an agent to be positive or negative. So, as discussed before, we characterise the influence using $\tilde\Theta$ wherein the effect of each agent $j$ on an agent $i$ is given by $|\theta_{ij}|$. We know from eqn. \eqref{eq:inf_x0} that the \textit{absolute} change brought about in this process is $\|\tilde\Theta(\mathbf{x}^i(0)-\mathbf{x}(0))\|_1=[|\Theta_{1i}||\delta| ~ |\Theta_{2i}||\delta| ~ ... ~ |\Theta_{ni}||\delta|]^T$. For a unit change ($\delta=1$) in agent $i$'s initial condition, the overall change in the network is given by $\sum_{k=1}^n |\Theta_{ki}|$. Therefore, the absolute centrality vector for all the agents becomes $\tilde\Theta\mathbb{1}_n$. Hence, proved.
\end{proof}

The above formulation allows us to determine the most influential agent which is simply the one with the highest \textit{absolute influence centrality}. { As opposed to IC, the proposed centrality measure is applicable for signed networks and accounts for the influence of non-stubborn agents as well.} 

\begin{expm}
\label{ex:centrality}
We define the matrix $\Theta$ for the network topology in Fig. \ref{fig:Network_with_cycles} whose influence coefficients were determined in Example \ref{expm:4}. 
$\Theta$ takes the following form:
\begin{align*}
   \Theta= \left[ \begin{smallmatrix}
    0.5 & 0  & 0  & 0 & 0.02 & 0.32 & 0 & 0& 0.023 & -0.0251 & -0.0314 \\
    0 & 0 & 0 & 0 &  0.2 &  0.2 & 0 & 0  & -0.058 & 0.0628 &  0.0784 \\
    0 & 0 & 0 & 0 &  0.2 &  0.2 & 0 & 0  & -0.058 & 0.0628 &  0.0784 \\
    0 & 0 & 0 & 0 &  0.2 &  0.2 & 0 & 0  & -0.058 & 0.0628 &  0.0784 \\
    0 & 0 & 0 & 0 & 1 & 0 & 0 & 0 & 0 & 0 & 0 \\
    0 & 0 & 0 & 0 & 0 & 1 & 0 &0 & 0 & 0 & 0 \\
    0 & 0 & 0 & 0 & 0 & 1 & 0 &0 & 0 & 0 & 0 \\
    0 & 0 & 0 & 0 & 0 & 1 & 0 &0 & 0 & 0 & 0 \\
    0 & 0 & 0 & 0 & 0 & 0 & 0 &0 &0.29 & -0.31 & -0.39  \\
    0 & 0 & 0 & 0 & 0 & 0 & 0 &0 &0.29 & -0.31 & -0.39  \\
    0 & 0 & 0 & 0 & 0 & 0 & 0 &0 &0.29 & -0.31 & -0.39  
\end{smallmatrix} \right]
\end{align*}

The absolute centrality for the network, as defined in Theorem \ref{thm:centrality}, is $[0.5, 0,0,0,1.62,3.92,0,0,1.08,1.15,1.44]$. It shows that the stubborn opinion leader $6$ is the most influential agent in the network, followed by non-stubborn opinion leaders $5,11,10,9$ and stubborn follower $1$. Since the opinion leader $6$ in the sink $S_2$ is stubborn, we observe that the non-stubborn opinion leaders $7$ and $8$ in $S_2$ have zero influence. Interestingly, stubbornness has enabled the followed $1$  to become influential; it now ranks right after the opinion leader $9$ in terms of its absolute centrality score. 
\end{expm}
\section{Conclusions}
\label{sec:con}
%
%

The paper quantifies the influence of the opinion leaders and the stubborn agents (if any) on the rest of the agents in a network $\mathcal{G}$ with signed interactions. 
We examine the complex interplay of signed interactions, the network topology and stubbornness in the opinion formation process. The underlying convergence analysis reveals that a non-stubborn opinion leader connected to a stubborn opinion leader cannot be influential even to its own followers. Using this analysis, the work posits the conditions an opinion leader must satisfy to be influential in the presence of stubbornness. Based on the network topology, we lay down the rules to construct the SFGs associated with $\mathcal{G}$. Thereafter, the precise values of the influence distributed over the network is calculated by using the Mason's gain formula. Using the SFGs, we establish that the gains of paths and the loops that form the underlying network determine the degree of influence of an influential agent on the rest. This finding is illustrated and verified using examples in the paper.

Further, we determine the most influential agent in the network as the one that results in the maximum deviation in the final opinion vector for a unit deviation in the agent's initial opinion. This leads to a new centrality measure, called the \textit{absolute influence centrality}, which allows us find the overall influence of an agent in the network. To the best of our knowledge, the proposed centrality measure is the first that accounts for the impact of stubborn behaviour and opinion leaders simultaneously along with signed interactions in the network. In future, we plan to analyse the suitable modifications of a network that can alter the influence of certain chosen agents in a desired manner. 

\section{Appendix}

\subsection{Rules for construction of SFGs}
We have constructed $\mathcal{G}_s$ from $\mathcal{G}$. 
Next, we use the following properties of agents in $\mathcal{G}$ to reduce the SFG $\mathcal{G}_s$ to $\bar{\mathcal{G}}_s$.
\begin{lemma}[\cite{degroot1974reaching},\cite{xia2015structural}]
\label{Thm:3}
    Consider a group of $n$ agents connected over a weakly connected digraph $\mathcal{G}$ whose opinions evolve according to eqn. \eqref{eq:opinion_dynamics}. {  If an  opinion leader $i$ is associated with a sink $S_l$ which does not contain any stubborn opinion leader}, then its opinion eventually converges to:
    \begin{itemize}
       \item $x_i(0) \ \forall \ i \in \mathcal{V}_{o1}$,
           \item $(\mathbf{w}_{P}^{l+1})^T \mathbf{x}_{l+1}(0) \ \forall \ i \in \mathcal{V}_{o2}$ such that $S_l\in \mathcal{S}_{cp}$, 
             \item $\sigma_i(\mathbf{w}_{P}^{l+1})^T \mathbf{x}_{l+1}(0) \ \forall \ i \in \mathcal{V}_{o2}$ such that $S_l\in \mathcal{S}_{bal}$.
            \item $0 \ \forall \ i \in \mathcal{V}_{o2}$ such that $S_l\in \mathcal{S}_{unbal}$
        \end{itemize}
            where $i\in \mathcal{V}_o$, $l \in \{1,...,n_s\}$ and { $\sigma_i \in \{1,-1\}$} depending on the partition of $\mathcal{V}_{lq}$, $q \in \{1,2\}$ to which $i$ belongs.
\end{lemma}

Using this result we construct the nodes of the reduced SFG $\bar{\mathcal{G}}_s$ from $\mathcal{G}_s$ as given below:
First, we consider a node $i$ in $\mathcal{G}_s$ for $i \in \{1,...,m\}$ associated with state $y_i$. Due to the numbering of nodes in $\mathcal{G}$, this node is associated with the final opinion of a follower $i$ in $\mathcal{G}$. It follows from eqn. \eqref{eqn:steady_state_OD} that the final opinion of a follower depends on the opinions of its neighbours and$/$or its initial opinion if it is stubborn. 
Therefore, node $i$ in $\mathcal{G}_s$ has at least one incoming branch originating from a node $j$ associated with the final opinion of $i's$ neighbour in $\mathcal{G}$ (an opinion leader or follower) or $i's$ initial opinion (if $i$ is stubborn). Consequently, node $i$ is not a source node in $\mathcal{G}_s$. Similarly, it continues to be a non-source node in $\bar{\mathcal{G}}_s$ represented as $\bar i$ associated with state $\bar{y}_i=y_i$. 

The node $i$ in $\mathcal{G}_s$ is associated with the final opinion of an opinion leader in $\mathcal{G}$ for $i\in\{m+1,...n\}$. If the opinion leader associated with $i$ is non-stubborn, we know from eqn. \eqref{eqn:steady_state_OD} that its final opinion depends only on the final opinion of opinion leaders in its corresponding sink $S_l=\mathcal{S}^v(i)$. For an opinion leader in $\mathcal{V}_{o1}$, its final opinion is independent of any other agent in $\mathcal{G}$. It implies that if $i$ is associated with an opinion leader in $\mathcal{V}_{o1}$, it forms a source in $\mathcal{G}_s$.
Equivalently, node $i$ in $\mathcal{G}_s$ associated with final opinion of opinion leader in $\mathcal{V}_{o1}$ forms a source 
$\mathcal{O}_r$ in $\bar{\mathcal{G}}_s$ associated with state $\bar{y}_{\mathcal{O}_r}=\bar{y}_i=x_i(0)$ {  for $r\in \{1,...,|\mathcal{V}_{o1}|\}$.}

If node $i$ is associated with an opinion leader in $\mathcal{V}_{o2}$, it has incoming branches in $\mathcal{G}_s$ from one or more nodes associated with opinion leaders in the same sink. Thus, it does not form a source in $\mathcal{G}_s$. However, as we noted in Remark \ref{rem:SFG_Red}, if node $i$ was a source, we could determine its effect on the opinions of the rest of the agents using Mason's gain formula \eqref{eq:sfg-gain} in the SFG. 
{  The following steps are taken in a direction to form source(s) in $\bar{\mathcal{G}}_s$ from the nodes in $\mathcal{G}_s$ associated with opinion leaders in $\mathcal{V}_{o2}$.

We know from Lemma \ref{Thm:3} that the non-stubborn opinion leaders that form sink $S_l$ in $C(\mathcal{G})$ achieve consensus when $S_l$ belongs in $\mathcal{S}_{cp}$.}
Consider a follower $i \in \mathcal{V}_{F}$ in $\mathcal{G}$ which has opinion leaders associated with sink $S_l$ in its neighbourhood such that sink $S_l \in \mathcal{S}_{cp} \cap \mathcal{S}_n$. Then, eqn. \eqref{eqn:y} becomes ${ y}_i=\sum_{h \in N_i}b_{ih}{ y}_h+{ \beta_i x_i(0)}=\sum_{h \in (N_i\setminus S_l)}b_{ih}{  y}_h+{  \sum_{t \in S_l}b_{it}y_t+\beta_i x_i(0)}$. As ${  y}_{  t}={  y_u=(\mathbf{w}_{P}^{l+1})^T \mathbf{x}_{l+1}(0)}$ for each $t,u$ $ \in S_l$, 
\begin{align}
\label{eqn:linearity}
 {  y}_i=\sum_{h \in (N_i\setminus S_l)}b_{ih}{  y}_h+{ \big(\sum_{t \in S_l}b_{it}\big)(\mathbf{w}_{P}^{l+1})^T \mathbf{x}_{l+1}(0)+\beta_i x_i(0)}   
\end{align}
 Further, we know that the final opinion of the opinion leaders in $S_l$
are independent of any other agent's opinion except those in $S_l$. Therefore, it follows from the 
 discussion that all the nodes in $\mathcal{G}_s$ associated with the final opinions of opinion leaders in $S_l$ can be represented by a single source $\mathcal{O}_r$ in $\bar{\mathcal{G}}_s$ associated with state $\bar{y}_{\mathcal{O}_r}$ equal to the consensus value $\bar{y}_{\mathcal{O}_r}=(\mathbf{w}_{P}^{l+1})^T \mathbf{x}_{l+1}(0)$ for $r\in\{|\mathcal{V}_{o1}|+1,...,|\mathcal{V}_{o1}|+|\mathcal{S}_{cp}|\}$.

Similarly, if a node $i$ in $\mathcal{G}_s$ is associated with the opinion leader in a sink $S_l\in \mathcal{S}_{bal} \cap \mathcal{S}_n$, then $S_l$ is { composed of non-stubborn opinion leaders} such that the opinion leaders in $S_l$ are partitioned into two sets $\mathcal{V}_{l1}$ and $\mathcal{V}_{l2}$ such that $\mathcal{V}_{lq} \neq \emptyset, (q\in\{1,2\})$. We know from Lemma \ref{Thm:3} that the opinion leaders within a partition have consensus and the final opinions of opinion leaders in the two different partitions are equal in magnitude and opposite in signs. 
Since the magnitude of the final opinion depends only on the initial opinions of opinion leaders in $S_l$, following a similar analysis preceding eqn. \eqref{eqn:linearity}, it follows that the opinion leaders in sink $S_l$
contribute two sources { $\mathcal{O}_r$ and $\mathcal{O}_{r+1}$} to $\bar{\mathcal{G}}_s$ corresponding to each of the bipartitions associated with states $\bar{y}_{\mathcal{O}_r}=a$ and $\bar{y}_{\mathcal{O}_{r+1}}=-a$, respectively { where $a=(\mathbf{w}_{P}^{l+1})^T \mathbf{x}_{l+1}(0)$ and $r \in \{|\mathcal{V}_{o1}|+|\mathcal{S}_{cp}|+1,...,|\mathcal{V}_{o1}|+|\mathcal{S}_{cp}|+2|\mathcal{S}_{bal}|-1\}$.}

{ Next, consider a sink $S_l\in\mathcal{S}_{cp} \cup \mathcal{S}_{bal} \cup \mathcal{S}_{unbal}$ which consists of one or more stubborn opinion leaders. Due to strong connectivity among opinion leaders in $S_l$, a node $i$ in $\mathcal{G}_s$ associated with the final opinion of an opinion leader in $S_l$ has an incoming branch or a path in $\mathcal{G}_s$ from the node associated with the initial opinion of the stubborn opinion leader in $S_l$. Therefore, $i$
is a non-source node in ${\mathcal{G}}_s$. Because opinion leaders in sink $S_l$ often disagree due to stubbornness, node $i$ contributes a non-source node $\bar i$ in $\bar{\mathcal{G}}_s$. }

{ If node $i$ in $\mathcal{G}_s$ is associated with an opinion leader in a sink $S_l\in \mathcal{S}_{unbal}$ such that
only non-stubborn opinion leaders in $\mathcal{G}$ constitute $S_l$, then it follows from Lemma \ref{Thm:3} that the opinions of all the corresponding opinion leaders converge to zero. Substituting this result in eqn. \eqref{eqn:y} leads to the conclusion that their influence on their followers is null. So, even if the network has such opinion leaders, they do not form any source in $\bar{\mathcal{G}}_s$. Since they neither influence nor are influenced by any other source, they do not form a node in $\bar{\mathcal{G}}_s$ 

Finally, a node $i$ in $\mathcal{G}_s$ for $i \in \{n+1,...,n+s\}$ corresponds to the initial opinion of a stubborn agent. The initial opinion of a stubborn agent contributes to the final opinions of both the stubborn agent and those with a path to it.  Since the initial opinion does not change, this node in $\mathcal{G}_s$ forms a source $\mathcal{O}_r$ in $\bar{\mathcal{G}}_s$ for $r \in \{|\mathcal{V}_{o1}|+|\mathcal{S}_{cp}|+2|\mathcal{S}_{bal}|+1,...,|\mathcal{V}_{o1}|+|\mathcal{S}_{cp}|+2|\mathcal{S}_{bal}|+s\}$. Since all the nodes in $\mathcal{G}_s$ are mapped to $\bar{\mathcal{G}}_s$, we derive the branches of $\bar{\mathcal{G}}_s$.
} 

We begin the determination of branches in $\bar{\mathcal{G}}_s$ by evaluating the branch gain of branch $(\mathcal{O}_r, \bar i)$ from source $\mathcal{O}_r$ to a non-source node $\bar i$. When $\mathcal{O}_r$ is a collective source formed by the opinion leaders in $S_l$, it follows from eqn. \eqref{eqn:linearity},
\begin{itemize}
    \item $g_{\bar i, \mathcal{O}_r}=\sum_{h \in S_l}b_{ih}$ when $S_l \in \mathcal{S}_{cp} \cap \mathcal{S}_n$
    \item $g_{\bar i, \mathcal{O}_r}=\sum_{h \in \mathcal{V}_{lq}}b_{ih}$ when $S_l \in \mathcal{S}_{bal} \cap \mathcal{S}_n$ and $\mathcal{O}_r$ corresponds to opinion leaders in $\mathcal{V}_{lq}$ for $q \in \{1,2\}$.
\end{itemize}
The gains of the remaining branches in $\bar{\mathcal{G}}_s$ can be determined using eqn. \eqref{eqn:y}.

\subsection{Proof of Theorem \ref{Thm:4}}
\begin{proof}
 The final opinions of node $\bar i$ in $\bar{\mathcal{G}}_s$ by using eqn. \eqref{eq:sfg-gain} can be formulated as,
\begin{align*}
\label{eqn:CII_verify}
   \bar{y}_i=\sum_{r=1}^{|\Omega_i|}c_{ir}{\bar y}_{\mathcal{O}_r} 
\end{align*}
where $\Omega_i$ is the set of sources in $\bar{\mathcal{G}}_s$ who have a forward path to node $\bar i$ and ${\bar y}_{\mathcal{O}_r}\in\mathbb{R}$ is the node state associated with the source $\mathcal{O}_r$. The individual influence coefficient $\theta_{ij}$ for pair of agents $i$ and $j$ in $\mathcal{G}$ 
is determined as follows: 
\begin{itemize}
\item If an opinion leader $j$ in $\mathcal{V}_{o1}$ or {  the initial opinion of a stubborn agent $j$ forms a source $\mathcal{O}_r$, then $\bar{y}_{\mathcal{O}_r}=x_j(0)$. So, $c_{ir}$ gives the influence of agent $j$ only. Thus,  $\theta_{ij}=c_{ir}$.}
\item The opinion leaders in sink $S_l\in \mathcal{S}_{cp} \cap \mathcal{S}_n$ form a source $\mathcal{O}_r$, 
we know that ${\bar{y}}_{\mathcal{O}_r}=(\mathbf{w}_{P}^{l+1})^T\mathbf{x}_l(0)$.  
 The contribution of opinion leader $j's$ initial opinion in consensus value is $(\mathbf{w}_{P}^{l+1})_{\kappa}x_j(0)$ when $j \in S_l$.
 Thus, $\theta_{ij}$ for opinion leader $j$ and agent
$i$ is $\theta_{ij}=c_{ir}(\mathbf{w}_{P}^{l+1})_{\kappa}$.

    \item The opinion leaders in sink $S_l \in \mathcal{S}_{bal} \cap \mathcal{S}_n$ form sources $\mathcal{O}_{r}$ and $\mathcal{O}_{r+1}$ depending on partitions $V_{lp}, \ p \in \{1,2\}$. The sources $\mathcal{O}_r$ and $\mathcal{O}_{r+1}$
    are associated with states $\sigma_{p}(\mathbf{w}_{P}^{l+1})^T\mathbf{x}_l(0)$ where $\sigma_p \in \{+1,-1\}, p \in \{1,2\}$. Note $\sigma_1 \neq \sigma_2$.
    
    An opinion leader $j \in S_l$ contributes to the state of both sources $(\mathbf{w}_{P}^{l+1})_{\kappa}x_j(0)$ and $-(\mathbf{w}_{P}^{l+1})_{\kappa}x_j(0)$, respectively.
    Thus, the exact influence of an opinion leader $j \in S_l$ on the follower $i$ is $\theta_{ij}=(\sigma_{1}c_{ir}+\sigma_2c_{i(r+1)})(\mathbf{w}_{P}^{l+1})_{\kappa}$.   
\end{itemize}
The remaining cases follow from Lemma \ref{Thm:3}. 
Hence, proved.
\end{proof}
\section*{References}
\bibliographystyle{IEEEtran}
\bibliography{IEEEabrv,references_2}
\vskip -2\baselineskip plus -1fil
\begin{IEEEbiography}[{\includegraphics[width=1in,clip,keepaspectratio]{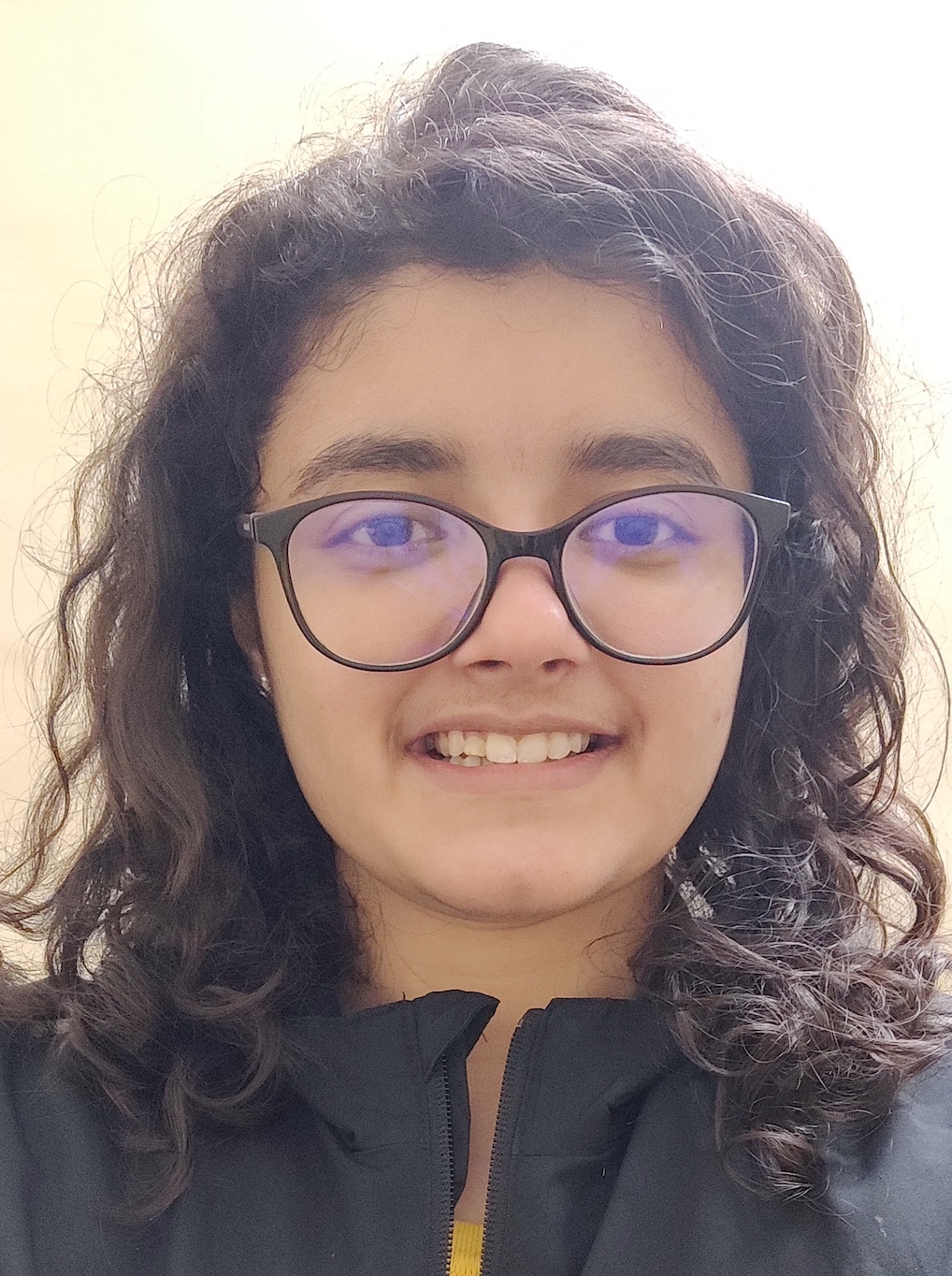}}]{Aashi Shrinate} received her B. Tech degree in Electrical Engineering from Motilal Nehru National Institute of Technology, Allahabad, India in 2020. She is working towards a PhD degree in control and automation specialization in Distributed Control and Decision Lab, Department of Electrical Engineering at the Indian Institute of Technology, Kanpur. She has been receiving the Prime Minister Research Fellowship from 2023. Her research focuses on networked dynamical systems with applications to opinion dynamics in social networks and robotic networks. 
\end{IEEEbiography}
\vskip -2\baselineskip plus -1fil
\begin{IEEEbiography}[{\includegraphics[width=1in,height=1.05in,clip]{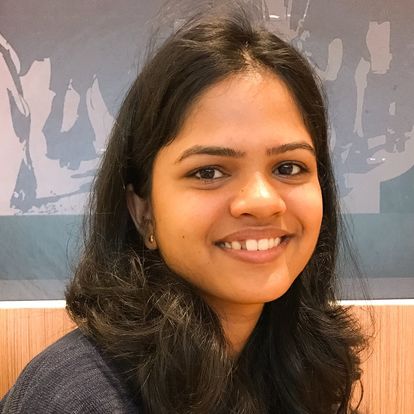}}]{Twinkle Tripathy} is currently an Assistant Professor in the Department of Electrical Engineering of IIT Kanpur. She received a Dual Degree of MTech. and Ph.D. at Systems \& Control Engineering, IIT Bombay in Dec. 2016. She started her post-doctoral tenure at the School of Electrical \& Electronic Engineering, NTU, Singapore. After serving there for a year, she joined the Faculty of Aerospace Engineering, Technion – Israel Institute of Technology as a post-doctoral fellow. Her research interests broadly include control and guidance of autonomous systems, cyclic pursuit strategies and opinion dynamics.
\end{IEEEbiography}
\end{document}